\documentclass[a4paper]{article}

\usepackage[utf8]{inputenc}

\usepackage{amssymb}
\usepackage{amsfonts}
\usepackage{amsmath}
\usepackage{amsthm}
\usepackage{hyperref}
\usepackage{enumitem}
\usepackage[ruled,lined,linesnumbered,noend]{algorithm2e}
\usepackage{url}
\usepackage{stmaryrd}
\usepackage{xcolor} 
\usepackage{authblk}
\usepackage[left=2.5cm,right=2.5cm,top=3cm,bottom=3cm]{geometry}
\usepackage{graphicx}
\usepackage{forest}
\usepackage[noadjust]{cite}

\usepackage{stmaryrd}
\usepackage{tensor}
\usepackage{cleveref}

\theoremstyle{plain}
\newtheorem{theorem}{Theorem}[section]
\newtheorem{proposition}[theorem]{Proposition}
\newtheorem{lemma}[theorem]{Lemma}

\newtheorem{corollary}[theorem]{Corollary}

\newtheorem{observation}[theorem]{Observation}

\theoremstyle{remark}
\newtheorem{remark}[theorem]{Remark}

\newcommand{\R}{\mathbb{R}}
\newcommand{\Q}{\mathbb{Q}}
\newcommand{\Z}{\mathbb{Z}}
\newcommand{\plusgg}{+_{\group'}}
\newcommand{\zerogg}{0_{\group'}}
\newcommand{\leqgg}{\leq_{\group'}}

\newcommand{\lessg}{<_{\group}}
\newcommand{\greatg}{>_{\group}}

\newcommand{\lessa}{<_{\Omega}}
\newcommand{\leqlex}{\leq_{\text{lex}}}

\newcommand{\arrange}{\mathcal{A}}
\newcommand{\sign}{\mathsf{sign}}
\newcommand{\infer}{\mathsf{infer}}
\newcommand{\supp}{\mathsf{supp}}

\newcommand{\monoid}{\mathbb{M}}
\newcommand{\monoidalphabet}{\Sigma_\monoid}
\newcommand{\concatm}{\circ_{\monoid}}

\newcommand{\group}{\mathbb{G}}
\newcommand{\plusg}{+_{\group}}
\newcommand{\leqg}{\leq_{\group}}
\newcommand{\zerog}{0_{\group}}

\DeclareMathOperator{\sourceName}{\mathsf{source}}			
\DeclareMathOperator{\targetName}{\mathsf{target}}			
\newcommand{\sourcefunc}[1]{\sourceName(#1)}	
\newcommand{\targetfunc}[1]{\targetName(#1)}

\DeclareMathOperator{\weightName}{\mathsf{weight}}			
\newcommand{\weightfunc}[1]{\weightName(#1)}	
\newcommand{\weightnumber}[1]{\mathsf{\#weights}(#1)}	
\DeclareMathOperator{\labelName}{\mathsf{label}}			
\newcommand{\labelfunc}[1]{\labelName(#1)}
\DeclareMathOperator{\bigO}{O}
\DeclareMathOperator{\Oh}{O}

\DeclareMathOperator{\bigo}{O}
\DeclareMathOperator{\Markers}{\mathcal{M}}
\DeclareMathOperator{\emptyMarker}{\diamond}
\newcommand{\doc}{\mathsf{D}}
\newcommand{\outputfunc}[1]{\mathsf{out}(#1)}

\DeclareMathOperator{\enumProbTransducer}{\textsf{CT-Enum}}
\DeclareMathOperator{\enumProbGraph}{\textsf{SP-Enum}}

\DeclareMathOperator{\poly}{poly}
\DeclareMathOperator{\eword}{\varepsilon}

\newcommand{\ta}{\ensuremath{\mathtt{a}}}
\newcommand{\tb}{\ensuremath{\mathtt{b}}}
\newcommand{\tc}{\ensuremath{\mathtt{c}}}

\newcommand{\varset}{\ensuremath{\mathcal{X}}}

\newcommand{\open}[1]{\tensor[^{}]{\langle}{_{#1}}}
\newcommand{\close}[1]{\tensor[_{#1}]{\rangle}{^{}}}
\DeclareMathOperator{\powerset}{\mathcal{P}}

\newcommand{\varsx}{\ensuremath{\mathsf{x}}}
\newcommand{\varsy}{\ensuremath{\mathsf{y}}}
\newcommand{\varsz}{\ensuremath{\mathsf{z}}}

\begin{document}

\title{Revisiting Weighted Information Extraction: A Simpler and Faster Algorithm for Ranked Enumeration}
\date{}

\author[1]{Pawe\l{} Gawrychowski}
\author[2]{Florin Manea}
\author[3]{Markus L. Schmid}

\affil[1]{University of Wroc\l{}aw, Wroc\l{}aw, Poland, \texttt{gawry@cs.uni.wroc.pl}}
\affil[2]{Computer Science Department and CIDAS, Universität Göttingen, Göttingen, Germany, \texttt{florin.manea@informatik.uni-goettingen.de}}
\affil[3]{Humboldt-Universit\"at zu Berlin, Unter den Linden 6, D-10099, Berlin, Germany, \texttt{MLSchmid@MLSchmid.de}}

\maketitle

\begin{abstract}
Information extraction from textual data, where the query is represented by a finite transducer and the task is to enumerate all results without repetition, and its extension to the weighted case, where each output element has a weight and the output elements are to be enumerated sorted by their weights, are important and well studied problems in database theory. On the one hand, the first framework already covers the well-known case of regular document spanners, while the latter setting covers several practically relevant tasks that cannot be described in the unweighted setting. 

It is known that in the unweighted case this problem can be solved with linear time preprocessing $O(|D|)$ and output-linear delay $O(|s|)$ in data complexity, where $D$ is the input data and $s$ is the current output element. For the weighted case, Bourhis, Grez, Jachiet, and Riveros [ICDT 2021] recently designed an algorithm with linear time preprocessing, but the delay of $O(|s| \cdot \log|\doc|)$ depends on the size of the data.

We first show how to leverage the existing results on enumerating shortest paths to obtain a simple alternative algorithm with linear preprocessing and a delay of $O(|s_i| + \min\{ \log i, \log|\doc|\})$ for the $i^{\text{th}}$ output element $s_i$ (in data complexity); thus, substantially improving the previous algorithm. Next, we develop a technically involved rounding technique that allows us to devise an algorithm with linear time preprocessing and output-linear delay $O(|s|)$ with high probability. To this end, we combine tools from algebra, high-dimensional geometry, and linear programming.
\end{abstract}

\section{Introduction}\label{sec:intro}

The term \emph{information extraction} usually refers to the task of compiling structured information from text documents. Its most famous instance is the framework of so-called \emph{document spanners} introduced in the seminal paper~\cite{FaginEtAl2015} (see the surveys~\cite{SchmidSchweikardt2022,DBLP:journals/sigmod/AmarilliBMN20, Schmid2024} or~\cite{BourhisEtAl2021, DoleschalEtAl2023, FreydenbergerThompson2022, MunozRiveros2023} for some recent publications on document spanners). A document spanner (over a set $\varset$ of variables) is a function that maps a document $\doc$ over some alphabet $\Sigma$ to a finite table with a column for each variable from $\varset$ and the entries of the cells of the table are just pairs of positions of $\doc$. This can be illustrated as follows (where $\Sigma = \{\ta, \tb, \tc\}$ and $\varset = \{\varsx, \varsy, \varsz\}$):

\begin{center}
$\doc = \ta \tb \tb \ta \tb \tc \tc \ta \tb \tc$ \hspace{0.5cm} $\implies$ \hspace{0.5cm}
\begin{tabular}{|l|l|l|}\hline
$\varsx$ & $\varsy$ & $\varsz$\\\hline\hline
$(2, 5)$ & $(4, 7)$ & $(1, 10)$\\\hline
$(3, 5)$ & $(5, 8)$ & $(4, 7)$\\\hline
$(1, 3)$ & $(3, 10)$ & $(2, 4)$\\\hline
$\ldots$ & $\ldots$ & $\ldots$\\\hline
\end{tabular}
\end{center}

\smallskip

The pairs $(i, j)$ are called \emph{spans} and are interpreted as pointers to factors of $\doc$, e.\,g., $(4, 7)$ refers to $\ta \tb \tc \tc$, since this is the factor that starts at position $4$ and ends at position $7$. A row of the table is called a \emph{span tuple} and the whole table is called a \emph{span relation}. Spanners are a suitable formalisation of relevant information extraction tasks (see the introductions of the papers mentioned above). 

The most fundamental class of spanners is the class of regular spanners (see~\cite{FaginEtAl2015}), which can be described by finite automata or regular expressions that represent span tuples by using brackets $\open{\varsx} \ldots \close{\varsx}$ with $\varsx \in \varset$ for marking the start and end points of the spans in the input string. For example, 
\begin{equation*}
r = (\tc + \tb)^* \open{\varsx} (\ta + \tb)^* \close{\varsx} \ta^* \open{\varsy} \ta^* \close{\varsy} \tc^*  
\end{equation*}
is a regular expression that describes a regular spanner over $\varset = \{\varsx, \varsy\}$: From the input document $\tc \tb \tc \ta \tb \ta \ta \ta \tc$, it can extract the span tuple $((4, 6), (8, 8))$ since $\tc \tb \tc \open{\varsx} \ta \tb \ta \close{\varsx} \ta \open{\varsy} \ta \close{\varsy} \tc \in L(r)$, or the span tuple $((4, 5), (6, 8))$ since $\tc \tb \tc \open{\varsx} \ta \tb \close{\varsx} \open{\varsy} \ta \ta  \ta \close{\varsy} \tc \in L(r)$.

A regular spanner can also be described by a finite transducer that simply marks some of the input positions (i.\,e., it marks input positions with the special markers $\open{\varsx}$, $\close{\varsx}$, $\open{\varsy}$, $\close{\varsy}$, etc.). This suggests the more general information extraction framework of \emph{annotation transducers} that properly extends regular spanners and that has already been used in~\cite{MunozRiveros2023,BourhisEtAl2021}. In this work, we will also focus on annotation transducers; the case of regular spanners is therefore a special case of our results.

An annotation transducer is a classical finite automaton, but transitions are not only labelled by an input symbol from $\Sigma$, but also by a marker symbol from a set $\Markers$ of \emph{markers}, which includes the special \emph{empty marker} $\emptyMarker$. Consequently, a computation of an annotation transducer marks the letters of the input by elements from $\Markers$, so it outputs a string over the alphabet $\Markers \times \Sigma$ whose restriction to $\Sigma$ equals the input. The \emph{output set} consists then in all the outputs of accepting computations, but we represent them in a contracted way: marked input letters are represented by their positions and all letters marked with the empty marker are ignored. For example, an annotation transducer may mark an input document $\ta \tb \ta \ta \tb \tc$ as $(\emptyMarker, \ta) (\gamma, \tb) (\delta, \ta) (\emptyMarker, \ta) (\emptyMarker, \tb) (\delta, \tc)$, which means that the corresponding \emph{output tuple} is $((\gamma, 2), (\delta, 3), (\delta, 6))$.

The arguably most important computational problem for information extraction based on annotation transducers is that of producing the set of output tuples for a given annotation transducer $\mathcal{T}$ and a given input document $\doc$, i.\,e., the associated query evaluation problem. As usual for query evaluation problems, we consider this in the following enumeration perspective: After a preprocessing with a running time that is only linear in $|\doc|$, we wish to enumerate all the output tuples with an output-linear delay, i.\,e., the time needed to produce the next output tuple is linear in the size of this tuple, but independent from $|\doc|$. Measuring only in dependency on $|\doc|$ and neglecting the dependency on the query size $|\mathcal{T}|$ covers the plausible scenario where the data is rather large, while the query is human-readable and therefore comparatively small. This complexity measurement is standard in query evaluation problems and is also called \emph{data complexity}. In this setting, linear-time preprocessing and output-linear delay in data complexity is the best we can hope for. The fact that the size of the output tuples is not lower bounded in the size of the input data can make it challenging to enumerate them with output-linear delay.

It is well-known that this enumeration task can be solved with linear preprocessing and output-linear delay in data complexity. Phrased for regular spanners, this has been shown in~\cite{AmarilliEtAl2021, DBLP:journals/sigmod/AmarilliBMN20, FlorenzanoEtAl2020}, and, in the context of MSO-enumeration over trees (which properly extend strings), it has already been solved in~\cite{Bagan2006}. It also follows from the result of~\cite{MunozRiveros2023}, which extends the scenario to the SLP-compressed setting. 

We are interested in the extension to the weighted scenario, which we will discuss next.

\medskip
\noindent\textbf{Weighted Information Extraction and Ranked Enumeration.} 
One of the main motivations for the enumeration perspective when approaching query evaluation problems is that the result set is potentially very large (even exponentially), which means that a user has to wait a long time until the whole result set is materialised. Instead, the idea is to provide the user with the first elements of the result set as fast as possible (hence only linear dependency on $|\doc|$ of the preprocessing) and the possibility to receive a new element on demand as fast as possible (hence, a delay that is linear with respect to size of the next element, but independent from $|\doc|$). While this setting covers many practically relevant aspects, it ignores that most likely not all answers to a query are equally interesting, especially if there are many answers. Thus, in the worst case, the relevant answers may appear at the very end of the enumeration and the user again has to wait for a long time. The weighted case of information extraction is a natural way to deal with this issue: Every output tuple has a weight and the enumeration algorithm should produce all output tuples sorted by their weights.

The weighted case of information extraction has been investigated in~\cite{DoleschalEtAl2020, DoleschalEtAl2023, BourhisEtAl2021}. It can be easily formalised by simply adding weights to the model of annotation transducers. More precisely, a \emph{cost transducer} is an annotation transducer whose transitions are not only labelled by an input symbol and a marker, but also by a \emph{weight}, which is an element from an ordered abelian group\footnote{For simplicity, let us here assume that $\group = (\mathbb{N}, +, 0, \leq)$; formal definitions follow below.}~$\group$. In this setting, every accepting computation does not only produce an output tuple, but also allocates a weight to it (i.\,e., the sum of all the weights on the transitions of the computation). Now, every possible output tuple has a weight and, since these weights are ordered, the output tuples are ordered. 
As shown in~\cite{BourhisEtAl2021}, cost transducers cover the expressive power of cost functions expressed in monadic second order logic.

Consequently, we can talk about \emph{ranked enumeration}, which is the same task as described above, just that the output tuples must be produced in increasing order with respect to their weights (or decreasing order, which amounts to the same problem). The obvious question is of course, whether we can still achieve linear preprocessing and output-linear delay in this setting. 
Before we explain what is known about the problem of ranked information extraction and outline our respective contributions, let us discuss some relevant use-cases.

\medskip
\noindent\textbf{Use-Cases for Ranked Information Extraction.} We first discuss some application scenarios of ranked information extraction in the context of regular spanners.

Assume that we want to compile a table with columns ``name'', ``address'', ``country'', ``email'', ``age'', etc. from a text document that contains customer information, i.\,e., each row of the extracted table then represents a customer. A respective spanner would therefore have a variable per field ``name'', ``address'', etc. (we neglect the actual structure of the spanner, which anyway depends on the format of the data). However, our data might be incomplete in the sense that a customer might be missing the field ``country'' or ``age''. In this case, rather than ignoring this customer altogether, it seems more useful to still extract a span tuple for it, but with "country" or "age" being undefined. Hence, spanners with undefined variables is a practically relevant extension first considered in~\cite{MaturanaEtAl2018} and since then generally adopted to the model. 

Nevertheless, we might have a strong priority of complete span tuples over incomplete ones and we wish this to be reflected in our enumeration algorithm, i.\,e., it seems natural that we want to first get all complete tuples, then those with one undefined variable (possibly ordered with respect to some preferences with respect to the individual variables, e.\,g., a missing ``age'' field is less problematic than a missing ``email'' field), then those with two missing variables and so on. It is easy to see that this is covered by ranked enumeration based on cost transducers (i.\,e., by adding weights to transitions that start allocating a span to a variable). 

But also without undefined variables, there are interesting applications of document spanners in the weighted case. For example, we can use weights to force the enumeration of the customer entries to be alphabetically ordered by the "name" field (by simply adding weights to the transitions that read the first letter of the ``name'' field). We can also prioritise span tuples with long spans over span tuples with shorter ones by adding a negative weight to the transitions that read those parts of the input that are extracted by variables. Or we could prioritise span-tuples whose spans are not too far apart from each other by giving non-negative weights to the letters read between spans.

A particularly interesting application of weighted information extraction seems to be that it provides a way of handling uncertain data. Assume that our textual data is faulty in the sense that some letters might have been replaced by other letters. This might be due to typos caused by humans (e.\,g., if we are dealing with data entered by hand), or errors introduced by technical problems when storing or communicating the data, or due to the fact that the data results from scientific experiments and exhibits the typical margin of measurement errors. A user might have knowledge of the format of the data, but not where the actual errors occur. This would not be too bad, if it would only mean that we extract the errors as well, but, in fact, a single erroneous letter can also mean that our information extraction query does not extract anything anymore. For example, if some conversion of the data has changed all occurrences of the "@"-sign to the letter "a", then an annotation transducer that marks the factors of a text document that satisfy the format of email addresses would be quite useless, since it does not discover any email addresses anymore and therefore does not produce any result. 

In general, we can deal with this scenario by allowing the annotation transducer to make mistakes as follows. For every original transition for some letter $x \in \Sigma$, we simply add transitions for every other letter from $\Sigma$ as well (with the same source and origin state). The idea is of course that this modified annotation transducer can also extract email addresses with the "@"-sign swapped into an "a"-sign, by using one of the new alternative transitions for this faulty symbol. But the modified transducer is way too general and will extract a lot of rubbish. So we could just punish the newly added transitions with a weight of $1$, while all other original transitions have a weight of $0$. If we now enumerate in increasing order, then we first get all the original results (i.\,e., those extracted by the original transducer), then all results that the transducer can extract by correcting exactly one symbol (e.\,g., email addresses with "a" instead of "@"), then those with two corrections and so on. The same rubbish as before will be produced, but only at the end of the enumeration, which can be interrupted as soon as we get results with a number of corrections that is over some threshold. 

Obviously, in actual applications, we would fine-tune this setting by adding the mistake-transitions only for certain symbols that are known to be crucial, like the ``$@$''-sign for extracting email addresses, or by using more complicated weighing schemes, reflecting, e.g., some sort of similarity between the newly added transitions leaving a state and the already existing ones. 

\medskip
\noindent\textbf{Our Contribution.} The central computational task investigated in this work is as follows (formal definitions follow soon): Given a text document $\doc$ and an \emph{unambiguous}\footnote{This means that every output is produced by exactly one possible computation.} cost transducer $\mathcal{T}$, we wish to enumerate all output tuples in increasing order by their weights.

It is shown in~\cite{BourhisEtAl2021} that this can be achieved with a preprocessing of $\bigO(|\mathcal{T}| |\doc|)$ and every output tuple $s$ has delay $\bigO(|s| \cdot \log(|\mathcal{T}| |\doc|))$.

\begin{remark}\label{DCRemark}
Note that for the rest of the paper, we state all running times in dependency on $|\doc|$ and $|\mathcal{T}|$, i.\,e., we measure the running times of our algorithms in \emph{combined complexity}. In this way, it will be clear that the running times of our algorithms have only a polynomial dependency on the query size. However, we keep in mind that our objective is to find algorithms optimised towards the data complexity perspective, i.\,e., the paramount goal is a preprocessing time that is linear in $|\doc|$ and a delay that is linear in the size of the next enumerated element and independent from $|\doc|$. In particular, we do not optimise our algorithms regarding the dependency on the query size (other than making sure that this dependency is polynomial).
\end{remark}

So in~\cite{BourhisEtAl2021}, the dependency of the preprocessing on the data size is only linear, but, unfortunately, in the delay there is a logarithmic factor that depends on $|\doc|$. In big data scenarios this would significantly slow down the enumeration compared to the unweighted case. Given the high relevance of ranked information extraction, this is unsatisfactory, and brings the challenge of determining whether output-linear delay can be achieved, i.\,e., can we solve ranked information extraction in the same running times as in the unweighted case?

We approach the problem by first phrasing it in a simpler way as a problem of computing shortest paths of an edge-weighted DAG. This allows us to apply known algorithmic techniques for weighted DAGs, in particular Eppstein's DAG of shortest paths (see~\cite{Eppstein1998}). As a result, we obtain a simply structured enumeration algorithm with preprocessing of $\bigO(|\mathcal{T}| |\doc|)$ and, if $s_1, s_2, \ldots, s_k$ is the sorted enumeration of the output tuples, then the delay of $s_i$ is $\bigO(|s_i| + \min\{ \log i, \log(|\mathcal{T}| |\doc|)\})$ for every $i \in [k]$. On the one hand, this improves the enumeration algorithm of~\cite{BourhisEtAl2021} since the logarithmic term in the delay $\bigO(|s_i| + \min\{ \log i, \log(|\mathcal{T}| |\doc|)\})$ is only additive, while it is a factor in the delay $\bigO(|s_i| \cdot \log(|\mathcal{T}| |\doc|))$ of the enumeration algorithm of~\cite{BourhisEtAl2021}. Moreover, since the additive logarithmic term of the delay for the $i^{\text{th}}$ element is actually the minimum of $\log(|\mathcal{T}| |\doc|)$ and $\log i$, we know that this additive term is very small for the first elements of the enumeration and only then slowly increases to $\log(|\mathcal{T}| |\doc|)$ after we have seen $\Omega(|\doc||\mathcal{T}|)$ elements. This is particularly desirable for ranked enumeration, where we have a guarantee that the most important elements (according to the weights) are produced at the beginning of the enumeration. In particular, for application scenarios where we are only interested in, say, the first $25$ elements, the additive term is indeed constant and therefore negligible.

As our main result, we extend this algorithm in a non-trivial way in order to address the question of whether we can achieve output-linear delay. At first, it might seem that the logarithmic factor is unavoidable, as outputting the results in the sorted order is as hard as sorting, which does need a linear-logarithmic time in the comparison model.
However, we make the following observation. For fixed $\doc$ and $\mathcal{T}$, all the weights of output tuples have the form $\sum^t_{i = 1} \alpha_{i} g_{i}$, where $\alpha_{1},\alpha_{2},\ldots,\alpha_{t} \in \mathbb{N} \cup \{0\}$ and $g_{1},g_{2},\ldots,g_{t}$ are the group elements from $\group$ that occur on transition of the cost transducer, so in particular $t$ is bounded by the query size $|\mathcal{T}|$. Thus, while we may need to sort many weights, they cannot be all arbitrarily complicated.

At the heart of our output-linear algorithm lies a non-trivial rounding procedure, which allows us to show that we can sort any $N \geq |\doc|$ such sums as defined above in time $\Oh(\poly(|\mathcal{T}|) N)$, with high probability (see Lemma~\ref{mainSortingLemma}).
Formally, the procedure always terminates and returns the correct answer, and for any constant $c\in \mathbb{N}$ works in time $\Oh(\poly(|\mathcal{T}|) N)$ with
probability at least $1-1/N^{c}$.
Finally, our enumeration algorithm works as follows. In the preprocessing, we pre-compute the output tuples with ranks $1, 2, \ldots, |\doc|$ and we sort them according to their weights in linear time with high probability in the size of the input. Then, in the first stage of the enumeration phase, we output these pre-computed elements, which buys us enough time to pre-compute and sort the output tuples with ranks $|\doc| + 1, |\doc| + 2, \ldots, 2|\doc|$. In the next stage, we therefore output the tuples with ranks $|\doc| + 1, |\doc| + 2, \ldots, 2|\doc|$ and pre-compute and sort the tuples with ranks $2|\doc| + 1, 2|\doc| + 2, \ldots, 4|\doc|$, and so on. For the pre-computation of the elements, we use again Eppstein's DAG of shortest paths~\cite{Eppstein1998} already mentioned above and a result by Frederickson~\cite{Frederickson93}. This yields an algorithm with linear preprocessing and output-linear delay with high probability (see Theorem~\ref{advancedEnumTheorem}), i.\,e., an algorithm with the same running time guarantees as in the unweighted case with high probability.

\medskip
\noindent\textbf{Related Work.}
The connection between ranked enumeration problems in database theory and computing the shortest paths of a graph (and Eppstein's approach) has also successfully been explored in~\cite{TziavelisEtAl2022}. However,~\cite{TziavelisEtAl2022} is concerned with ranked enumeration of conjunctive queries and, moreover, the output paths are represented as explicit listings (i.e., the path size is proportional to the number of its edges) instead of the potentially much smaller labels, as in~\cite{BourhisEtAl2021} and as we also require in our setting.\par
In~\cite{AmarilliEtAl2024}, a work that is also related to~\cite{BourhisEtAl2021}, the authors consider ranked enumeration of MSO-queries over trees. However, the MSO-queries of~\cite{AmarilliEtAl2024} have only free FO-variables, but no free set-variables, which means that this setting does not extend MSO on strings as considered in our paper (or in~\cite{BourhisEtAl2021}). Moreover, the ranking function used in~\cite{AmarilliEtAl2024} requires some additional properties.

\section{Preliminaries}

\noindent\textbf{Algebra.} A \emph{free monoid} is a tuple $(\monoid, \concatm, \eword)$, where $\monoid = \monoidalphabet^*$ for a finite alphabet $\monoidalphabet$. For free monoids, we assume that each element from $\monoidalphabet$ can be stored in a single machine word, and that, for every $x, y \in \monoid$, calculating $x \concatm y$ takes constant time. 

A \emph{group} is a tuple $(\group, \plusg, \zerog)$, where $\group$ is a set of elements, $\plusg$ is an associative binary operation on $\group$, $\zerog \in \group$ is a neutral element for $\plusg$ (i.\,e., $\zerog \plusg x = x \plusg \zerog = x$ for every $x \in \group$) and every $x \in \group$ has an inverse $-x \in \group$ (i.\,e., $x \plusg (-x) = 0$). A group is \emph{abelian}, if $\plusg$ is also commutative. An \emph{ordered} group is a tuple $(\group, \plusg, \zerog, \leqg)$, where $(\group, \plusg, \zerog)$ is a group, and $\leqg$ is a total order over $\group$ that respects $\plusg$ (i.\,e., $x_1 \leqg x_2$ implies $x_1 \plusg y \leqg x_2 \plusg y$ for every $y \in \group$); for an ordered group, we denote by $|x|$ the greater element from $x$ and $-x$ with respect to the order $\leqg$. We will assume that $\group$ has an effective representation, meaning that every element of $\group$ is stored in a single machine word, calculating $x \plusg y$ takes constant time, and so does checking if $x \leqg y$.

\medskip
\noindent\textbf{Graphs.} 
Let $(\group, \plusg, \zerog, \leqg)$ be a fixed ordered abelian group, called the \emph{weight group}, and let $(\monoid, \concatm, \eword)$ be a fixed free monoid, called the \emph{label monoid}. 

We now define directed graphs with edges that have a weight from $\group$ and a label from $\monoid$. More formally, a \emph{directed graph with edge weights from $\group$ and edge labels from $\monoid$} is a tuple $G = (V, E, \weightName, \labelName)$, where $V$ is a set of nodes, $E$ is a set of edges, where, for every $e \in E$, $\sourcefunc{e} \in V$ is the \emph{source} of $w$ and $\targetfunc{e} \in V$ is the \emph{target} of $w$, the function $\weightName : E \to \group$ is the \emph{weight function} and $\labelName : E \to \monoid$ is the \emph{label function}.
For simplicity, we also represent an edge $e$ in the form $(\sourcefunc{e}, \weightfunc{e}, \labelfunc{e}, \targetfunc{e})$. A path in $G$ is a sequence of edges
\begin{equation*}
P = ((v_0, w_1, \beta_1, v_1), (v_1, w_2, \beta_2, v_2), \ldots, (v_{k-1}, w_k, \beta_k, v_k))\,.
\end{equation*}
The \emph{weight} of the path $P$ is defined by $\weightfunc{P} = \sum^k_{i = 1} w_i$ and the \emph{label} of the path $P$ is defined by $\labelfunc{P} = \prod^{k}_{i = 1} \beta_i$. We also denote the weight of an edge or a path as its \emph{length}, and by \emph{shortest} path, we mean a path of minimal length. For $u, v \in V$, a $u$-to-$v$-path is any path from $u$ to $v$.

We measure the size of a graph $G = (V, E, \weightName, \labelName)$ with edge weights from $\group$ and edge labels from $\monoid$ as $|G| = |V| + |E|$. This is justified, since every edge $e$ has constant size (due to the fact that $\weightfunc{e}$ and $\labelfunc{e}$ can be stored in single machine words).

\medskip
\noindent\textbf{Nondeterministic Finite Automata.} A \emph{nondeterministic finite automaton} (\emph{NFA} for short) is a tuple $M = (Q, \Sigma, \delta, q_0, F)$, where $Q$ is the finite set of \emph{states}, $\Sigma$ is the finite \emph{input alphabet}, $\delta : Q \times \Sigma \to \powerset(Q)$ is the \emph{transition function}, $q_0 \in Q$ is the \emph{initial state} and $F \subseteq Q$ is the set of \emph{final states}. The semantics are defined as usual, but we shall state them in detail for the transducer models defined later.

\medskip
\noindent\textbf{Enumeration Algorithms.} An \emph{enumeration} of a finite set $A$ is any sequence $(s_1, s_2, \ldots, s_m)$ such that $|A| = m$ and $A = \{s_1, s_2, \ldots, s_m\}$. A finite set $A$ is called \emph{$\group$-ranked}, if every $s \in A$ has a weight $\weightfunc{s} \in \group$. A \emph{ranked enumeration} of a $\group$-ranked set $A$ is any enumeration $(s_1, s_2, \ldots, s_m)$ of $A$ that satisfies $\weightfunc{s_1} \leqg \weightfunc{s_2} \leqg \ldots \leqg \weightfunc{s_m}$.

An \emph{enumeration problem} is a function $P$ that maps an input $I$ to a finite \emph{output set} $P(I)$, the elements of which are called \emph{output elements}. An enumeration problem $P$ is called \emph{$\group$-ranked}, if, for every input $I$, the output set $P(I)$ is a $\group$-ranked set.

An \emph{enumeration algorithm} for an enumeration problem $P$ is an algorithm $A$ that, on input $I$, produces an enumeration of $P(I)$. A \emph{ranked enumeration algorithm} for a $\group$-ranked enumeration problem $P$ is an algorithm $A$ that, on input $I$, produces a $\group$-ranked enumeration of $P(I)$. We assume that any enumeration algorithm first performs a \emph{preprocessing phase}, which is then followed by an \emph{enumeration phase}, in which the enumeration is produced.

Let $A$ be an enumeration algorithm for some enumeration problem $P$. The \emph{preprocessing time} of $A$ on input $I$ is the running-time of the preprocessing phase of $A$ on input $I$. If $(s_1, s_2, \ldots, s_m)$ is the output of $A$ on input $I$, then, for every $i \in [m]$, the \emph{delay of $s_i$} is the time that elapses between the end of completely producing the element $s_{i-1}$ (or the end of the preprocessing phase if $i=1$) and the end of completely producing the next output element $s_i$. The delay of $A$ on some input $I$ is the maximum over all delays of any output element. The preprocessing time and the delay of the algorithm $A$ is the maximum preprocessing time and maximum delay, respectively, over all possible inputs of length at most $n$ (viewed as a function of $n$).

Enumeration algorithms model the requirement that in practice the user cannot afford to wait until the set $P(I)$ is completely computed before receiving an answer. Instead, the user should be quickly provided with some answer and can then stop the enumeration as soon as enough answers have been seen. The ranked perspective even provides a way to postulate that the most important answers are enumerated first, which is, in most practical scenarios, a highly desirable feature.

However, in the case that the full output set is extremely large (i.\,e., exponential in the size of the queried data), an enumeration algorithm may have to deal with an exponential number of elements. This is not a problem if these output elements are just written on some output storage, but if the algorithm needs to perform computations on an exponential number of output elements, one runs into problems with respect to addressing and working with these elements. Clearly, this only happens if the output is indeed of exponential size, and in the very unlikely case that the user is really interested in seeing all of these exponentially many output elements and wants the algorithm to run all the way until the end. Nevertheless, on a formal level, we have to deal with this issue, which is why we also define polynomially-bounded versions of ranked enumeration problems, where the ranked enumeration stops after a polynomially number of elements have been produced.

Let $P$ be a $\group$-ranked enumeration problem and let $f : \mathbb{N} \to \mathbb{N}$ be some polynomial function. The \emph{$f$-bounded} version of $P$, denoted by $P^f$, is the restriction of $P$ to the case where, on input $I$, we only enumerate the $f(|I|)$ smallest  elements of $P(I)$ (with respect to their weights). Let us define this more formally. If $I$ is an input for $P$, where $P(I) = \{s_1, s_2, \ldots, s_m\}$ with $\weightfunc{s_i} \leqg \weightfunc{s_{i+1}}$ for every $i \in [m-1]$, then $P^f(I) = \{s_1, s_2, \ldots, s_{m'}\}$, where $m' = \min\{m, f(|I|)\}$. Note that $P^f$ is a $\group$-ranked enumeration problem in the sense defined at the beginning of this subsection.

\begin{remark}
In the following, we let $f$ be some fixed polynomial that we will use for polynomially-bounded versions of ranked enumeration problems. Our algorithms do not depend on $f$ and they will work for every choice of $f$. For example, taking $f$ to be linear covers those practical scenarios where we can assume that users never materialise a number of query answers that exceed the total size of the queried data itself; in the case of enumeration algorithms with linear time preprocessing, this models the case when the user allows roughly the same time for seeing the answers as for the preprocessing.
\end{remark}

\medskip
\noindent\textbf{Computational Model.}
The computational model we use is the standard unit-cost word RAM with word-size $\omega$ (meaning that each memory word can hold $\omega$ bits). It is assumed that this model allows processing inputs of size $n$, where $\omega \geq \log n$; in other words, the size $n$ of the data never exceeds (but, in the worst case, is equal to) $2^\omega$. Intuitively, the size of the memory word is determined by the processor, and larger inputs require a stronger processor (which can, of course, deal with much smaller inputs as well). Indirect addressing and basic arithmetical operations on such memory words are assumed to work in constant time. Note that numbers with $\ell$ bits are represented in $\bigO(\ell/\omega )$ memory words, and working with them takes time proportional to the number of memory words on which they are represented. This is a standard computational model for the analysis of algorithms, as defined in \cite{FredmanW90}, and used in classical works related to basic problems such as integer-sorting, see, e.g.,~\cite{AnderssonHNR98,HanT02,AnderssonT07}.

In the problem considered here, we assume that we are given a text document $\doc$, of length $n$, over some input alphabet $\Sigma$, and an unambiguous cost transducer $\mathcal{T}$ (to be formally defined in the next section). As explained in Section~\ref{sec:intro}, the runs of $\mathcal{T}$ on $\doc$ determine a set of weighted output tuples, which therefore defines a ranked enumeration problem. We wish to find ranked enumeration algorithms for this problem. 

As common in query evaluation tasks (and as also explained in Remark~\ref{DCRemark}), the dependency on the query size $|\mathcal{T}|$ is considered small in comparison to the data size $|\doc|$. This is why we are mainly concerned with running times that have a low dependency on the data size. More precisely, for the preprocessing phase, we aim for a running time with only linear dependency on $|\doc|$, while the delay should ideally be independent from $|\doc|$ (and only linear in the size of the output element). Assuming that we have to read the input data at least once, this is the best we can hope for with respect to the dependency on the data size. As far as the dependency on the query size is concerned, it is desired (although this is not a main concern) to still be  polynomial.

For dealing with the ranked setting, our first algorithm (Theorem~\ref{simpleEnumTransducerTheorem}) has to maintain a heap data structure whose size grows with the size of the set of output elements, while our second algorithm (Theorem~\ref{PathEnumWithSortingTheorem}) has to repeatedly sort a large number of output elements. In both cases, the space complexity can become exponential if we are dealing with instances that lead to an exponential number of output elements and if the enumeration phase is carried out completely. Since such pathological cases cannot be easily handled by the RAM model with logarithmic word size (e.g., in the case when the input size is roughly $2^\omega$, with $\omega$ being the word size), we only consider the $f$-bounded variant of our ranked enumeration problem (see the explanations from above). This is a mere technicality that is required to define our algorithms and respective running time guarantees in a sound way, using the word RAM model. In all reasonable application scenarios, it is to be assumed that ranked enumeration algorithms are not run until the end in the case that the output is exponentially large.

\section{Ranked Information Extraction}\label{sec:rankedInformationExtraction}

We next thoroughly formalise the enumeration problem investigated in this work, which has been discussed in the introduction on an intuitive level.

In the following, let $(\mathbb{G}, \plusg, \zerog, \leqg)$ be a fixed ordered abelian group. Let $\Sigma$ be a finite \emph{input alphabet} and let $\Markers$ be a finite set of \emph{markers}, including a designated \emph{empty marker} $\emptyMarker \in \Markers$. We assume that every element of $\Sigma$ and every element of $\Markers$ can be stored in a single machine word.\par 
Syntactically, a \emph{cost transducer over $\mathbb{G}$} is an NFA whose transitions are labelled by elements $(b, w, \gamma)$, where $b \in \Sigma$, $w \in \mathbb{G}$ and $\gamma \in \Markers$, i.\,e., an NFA $\mathcal{T} = (Q, (\Sigma \times \mathbb{G} \times \Markers), \delta, q_0, F)$. For convenience, we treat $\delta$ as a subset of $Q \times \Sigma \times \mathbb{G} \times \Markers \times Q$. The semantics are defined as follows.

A \emph{run} of $\mathcal{T}$ on some input $\doc \in \Sigma^*$ is a sequence $\rho = p_0 \to p_1 \to \ldots \to p_n$ of states $p_i \in Q$, $1 \leq i \leq n$, such that $|\doc| = n$ and, for every $i \in [n]$, there is a transition $(p_{i-1}, \doc[i], w_i, \gamma_i, p_i)$. We also write runs in the form
\begin{equation*}
\rho = p_0 \xrightarrow{\doc[1], w_1, \gamma_1} p_1 \xrightarrow{\doc[2], w_2, \gamma_2} \ldots \xrightarrow{\doc[n], w_n, \gamma_n} p_n\,.
\end{equation*}
The \emph{weight} of $\rho$ is $\weightfunc{\rho} = \sum^n_{i = 1} w_i$, while its \emph{output} is $\outputfunc{\rho} = (\gamma_{j_1}, j_1) (\gamma_{j_2}, j_2) \ldots (\gamma_{j_m}, j_m)$, where $1 \leq j_1 < j_2 < \ldots < j_m \leq n$ are exactly the positions of $\rho$ with $\gamma_{j_\ell} \neq \emptyMarker$ for every $\ell \in [m]$. If $p_0 = q_0$ and $p_n \in F$, then the run $\rho$ is \emph{accepting}. 

The cost transducer $\mathcal{T}$ is \emph{unambiguous}, if there are no two distinct accepting runs $\rho_1$ and $\rho_2$ on any input $\doc \in \Sigma^*$ such that $\outputfunc{\rho_1} = \outputfunc{\rho_2}$.

\begin{remark}
In the following, we assume that cost-transducers are unambiguous (an assumption also made in~\cite{BourhisEtAl2021}). This is a proper restriction of the model of cost-transducers, since general cost-transducer may have two runs on the same input that produce the same output, but with different weights. However, in our setting of information extraction, weights are used to allocate a priority to each output element, which in turn induces the desired order of the output elements. Consequently, the possibility to allocate several different weights to the same output element is not a useful feature. Note, however, that different output elements may have the same weights.
\end{remark}

Finally, the function $\llbracket \mathcal{T} \rrbracket : \Sigma^* \to \powerset((\Markers \times \mathbb{N})^*)$ is defined by
\begin{equation*}
\llbracket \mathcal{T} \rrbracket(\doc) = \{\outputfunc{\rho} \mid \text{$\rho$ is an accepting run of $\mathcal{T}$ on $\doc$}\}\,.
\end{equation*}

\begin{remark}
Note that for the sake of simplicity, we define cost transducers in a sightly different yet equivalent way compared to~\cite{BourhisEtAl2021}. The main purpose of cost transducers in~\cite{BourhisEtAl2021} is to express MSO-cost functions, which map set variables to sets of positions of the input string. Consequently, the cost transducers of~\cite{BourhisEtAl2021} mark every position of the input string by a whole set of output symbols (which encode the set variables that contain this position). For simplicity, we just use single markers and a designated empty marker that represents the empty set of output symbols (note that positions marked with the empty marker do not appear in the output, just like in the setting of~\cite{BourhisEtAl2021}, where the positions marked with the empty set of output symbols do not appear in the output). Another difference is that the model from~\cite{BourhisEtAl2021} can also allocate weights to the initial and final states. This is mere syntactical sugar, since by a simple modification of the underlying automaton, such initial and final weights can be moved onto the transitions leaving or entering the initial or final states, respectively.
\end{remark}

The problem $\enumProbTransducer$ is the ranked enumeration problem defined as follows. For a given unambiguous cost transducer $\mathcal{T}$ and a document $\doc \in \Sigma^*$, let $\enumProbTransducer(\mathcal{T}, \doc) = \llbracket \mathcal{T} \rrbracket(\doc)$. 

Bourhis et al.~\cite{BourhisEtAl2021} give a ranked enumeration algorithm for $\enumProbTransducer^f$. 
Worth noting, Bourhis et al.~\cite{BourhisEtAl2021} do not introduce their result in the context of bounded ranked enumeration, but their algorithm also seems to be using exponential space, in the worst case, when all the solutions to $\enumProbTransducer$ are output. Therefore, we state their result in our setting, for the sake of soundness and uniformity.

\begin{theorem}[Bourhis et al.~\cite{BourhisEtAl2021}]\label{BorhisResultTheorem}
There is a ranked enumeration algorithm for $\enumProbTransducer^f$ such that, on input $\mathcal{T}$ and $\doc$, the preprocessing time is $\bigO(|\mathcal{T}| |\doc|)$ and every output element $s$ has delay $\bigO(|s| \cdot \log(|\mathcal{T}| |\doc|))$.
\end{theorem}

Before we present our algorithmic approaches to the ranked enumeration problem $\enumProbTransducer^f$, we formulate it as a problem of enumerating shortest paths in an edge-labelled weighted graph. 

\medskip
\noindent\textbf{Reducing $\enumProbTransducer$ to Shortest Path Enumeration.} 
Let $\doc \in \Sigma^*$ with $|\doc| = n$ and let $\mathcal{T} = (Q, (\Sigma \times \mathbb{G} \times \Markers), \delta, q_0, F)$ be an unambiguous cost transducer. Without loss of generality, we can assume that $F = \{q_f\}$. Indeed, this can be achieved by first adding a new final state $q_f$ and making all former final states non-final. Then, for every transition $(p, b, w, \gamma, q)$ where $q$ is a former final state, we add a transition $(p, b, w, \gamma, q_f)$. Now we have a single final state $q_f$ and, for any input string, the cost transducer has the same set of outputs. Moreover, unambiguity is maintained: If there are two different accepting runs on the same input with the same output (i.\,e., unambiguity is violated), then we can replace the last transitions $(p_1, b, w_1, \gamma, q_f)$ and $(p_2, b, w_2, \gamma, q_f)$ of these runs by the transitions $(p_1, b, w_1, \gamma, q_1)$ and $(p_2, b, w_2, \gamma, q_2)$, where $q_1$ and $q_2$ are former final states; thus, unambiguity was already violated before the construction, which is a contradiction.

We define a monoid $(\monoid, \concatm, \eword)$, where $\monoid$ is the set of strings over the alphabet $(\Markers \setminus \{\emptyMarker\}) \times \{1, \ldots, n\}$ with the empty word being the neutral element, i.\,e., $\monoid = ((\Markers \setminus \{\emptyMarker\}) \times \{1, \ldots, n\})^*$, and the operation $\concatm$ is string-concatenation (i.\,e., $\eword$ plays the role of the empty string). 

Next, we define a directed graph $G_{\doc, \mathcal{T}} = (V_{\doc, \mathcal{T}},  E_{\doc, \mathcal{T}}, \weightName, \labelName)$ with edge weights from $\group$ and edge labels from $\monoid$ as follows. The set of nodes is $V_{\doc, \mathcal{T}} = \{(q, i) \mid q \in Q, 0 \leq i \leq n\}$. For every $i \in [n]$ and every transition $(p, \doc[i], w, \gamma, q) \in \delta$ with $\gamma \neq \emptyMarker$, there is an edge $((p, i-1), w, (\gamma, i), (q, i)) \in E_{\doc, \mathcal{T}}$, and for every transition $(p, \doc[i], w, \emptyMarker, q) \in \delta$, there is an edge $((p, i-1), w, \eword, (q, i)) \in E_{\doc, \mathcal{T}}$. 
It can be easily seen that there is a weight preserving one-to-one-correspondence between accepting runs of $\mathcal{T}$ on $\doc$ and $(q_0, 0)$-to-$(q_f, n)$-paths of $G_{\doc, \mathcal{T}}$. More precisely, there is a run 
\begin{equation*}
\rho = q_0 \xrightarrow{\doc[1], w_1, \gamma_1} q_1 \xrightarrow{\doc[2], w_2, \gamma_2} \ldots \xrightarrow{\doc[n], w_n, \gamma_n} q_n\,.
\end{equation*}
of $\mathcal{T}$ on $\doc$ with $q_n = q_f$ if and only if 
\begin{equation*}
P =\:((q_0, 0), w_1, \psi_1, (q_1, 1)), ((q_1, 1), w_2, \psi_2, (q_2, 2)), \ldots, ((q_{n-1}, n-1), w_n, \psi_n, (q_n, n)) 
\end{equation*}
is a path in $G_{\doc, \mathcal{T}}$ with $q_n = q_f$ and, for every $i \in \{1, 2, \ldots, n\}$, $\psi_i = (\gamma_i, i)$ if $\gamma_i \neq \emptyMarker$, and $\psi_i = \eword$ if $\gamma_i = \emptyMarker$. In particular, $\weightfunc{\rho} = \weightfunc{P}$ and $\outputfunc{\rho} = \labelfunc{P}$.\par

Since all edges $e \in E_{\doc, \mathcal{T}}$ satisfy $\sourcefunc{e} = (p, i-1)$ and $\targetfunc{e} = (q, i)$ for some $i \in [n]$, the graph $G_{\doc, \mathcal{T}}$ is a DAG. We call $G_{\doc, \mathcal{T}}$ the \emph{DAG-representation} of $\mathcal{T}$ and $\doc$.

\begin{observation}\label{sizeOfDAGObservation}
$|V_{\doc, \mathcal{T}}| = \bigo(|\doc||Q|)$, and every transition $(p, b, w, \gamma, q) \in \delta$ is responsible for at most $|\doc|$ edges of the form $((p, i-1), w, \psi, (q, i)) \in E_{\doc, \mathcal{T}}$ with $\doc[i] = b$, so, $|E_{\doc, \mathcal{T}}| = \bigo(|\doc||\mathcal{T}|)$. We conclude that $|G_{\doc, \mathcal{T}}| = \bigo(|\doc||\mathcal{T}|)$, and we can construct $G_{\doc, \mathcal{T}}$ in time $\bigo(|\doc||\mathcal{T}|)$.
\end{observation}

This means that the problem $\enumProbTransducer$ reduces to the following ranked enumeration problem $\enumProbGraph$: For a given DAG $G = (V,  E, \weightName, \labelName)$ with edge weights from $\group$ and edge labels from $\monoid$, and designated $s, t \in V$, let $\enumProbGraph(G, s, t) = \{\labelfunc{P} \mid P \text{ is an $s$-to-$t$-path of $G$}\}$. Note that the weight of an element $\labelfunc{P}$ of the output set is given by $\weightfunc{P}$.

Consequently, we now want to solve the ranked enumeration problem $\enumProbGraph^f$. We will next see (Theorem~\ref{simpleEnumTransducerTheorem}) that an algorithm for this problem can be obtained with moderate effort by using a data structure for the shortest paths of a DAG by Eppstein~\cite{Eppstein1998}, which we will discuss next.

\section{Eppstein's DAG of Shortest Paths}

Let $G = (V, E, \weightName, \labelName)$ be a graph with edge weights from $\group$ and edge labels from $\monoid$. Let $s, t \in V$ be distinguished nodes of $G$.\par
A \emph{heap representation of $G$'s $s$-to-$t$-paths} is a tree $\mathcal{H}(G)$. There is a one-to-one correspondence between the $s$-to-$t$-paths of $G$ and the nodes of $\mathcal{H}(G)$, and $\mathcal{H}(G)$ is a min-heap with respect to the weights of the $s$-to-$t$-paths, i.\,e., if a node $q$ of $\mathcal{H}(G)$ corresponds to an $s$-to-$t$-path $P_q$ and $q$ has a child $r$, then $r$ corresponds to an $s$-to-$t$-path $P_r$ with $\weightfunc{P_q} \leqg \weightfunc{P_r}$. 

The following result is shown in~\cite{Eppstein1998}, although here we phrase it in a way suitable for our applications. More details about this result can be found in Appendix~\ref{sec:AppendixEppstein}.

\begin{theorem}[Eppstein~\cite{Eppstein1998}]\label{EppsteinDAGTheorem}
Let $G = (V, E)$ be a DAG with edge weights from $\group$ and edge labels from $\monoid$, and let $s, t \in V$. We can compute in time and space $\bigO(|G|)$ a data structure that represents a heap representation $\mathcal{H}(G)$ of $G$'s $s$-to-$t$-paths. The degree of $\mathcal{H}(G)$ is $4$, we can move from a node $q$ of $\mathcal{H}(G)$ to any of its children in constant time, and, for every node $q$ of $\mathcal{H}(G)$ that represents an $s$-to-$t$-path $P_q$, we can retrieve $\weightfunc{P_q}$ in constant time and $\labelfunc{P_q}$ in time $\bigo(|\labelfunc{P_q}|)$.
\end{theorem}

Note that we can only traverse the heap $\mathcal{H}(G)$ starting from its root; in particular, $\mathcal{H}(G)$ does not support a remove-min operation. Nevertheless, we can use $\mathcal{H}(G)$ for a ranked enumeration algorithm for $\enumProbGraph^f$. The general idea is the same as how to obtain the $k$ smallest elements from a min-heap (that we can only traverse) in time $\bigo(k \log k)$ as also sketched in the introduction of~\cite{Frederickson93}: We always produce the smallest element of an auxiliary min-heap (that supports a remove-min operation) as output, where the auxiliary min-heap stores all nodes of $\mathcal{H}(G)$ that are children of nodes that have already been produced as output.
This allows us to obtain an enumeration algorithm with delay $\bigO(|s_{i}| + \log i)$. In fact, a simple tweak keeps the size of the auxiliary min-heap
upper bounded by $|G|$, making the delay $\bigO(|s_{i}| + \min\{\log i,\log |G|\})$.

\begin{corollary}\label{simpleEnumGraphCorollary}
There is a ranked enumeration algorithm for $\enumProbGraph^f$ such that, on input $(G, s, t)$, the preprocessing time is $\bigO(|G|)$ and, if $(s_1, s_2, \ldots, s_k)$ is the output of the algorithm, then the delay of $s_i$ is $\bigO(|s_i| + \min\{\log i, \log |G| \})$ for every $i \in [k]$. 
\end{corollary}

\begin{proof}
A ranked enumeration algorithm for $\enumProbGraph^f$ with the claimed complexity bounds can be obtained as follows. In the preprocessing phase, we construct a data structure that represents a heap representation $\mathcal{H}(G)$ of $G$'s $s$-to-$t$-paths with degree $4$, such that we can move from a node $q$ of $\mathcal{H}(G)$ to any of its children in constant time, and every node $q$ of $\mathcal{H}(G)$ that represents an $s$-to-$t$-path $P_q$ stores a pointer to $\labelfunc{P_q}$. According to Theorem~\ref{EppsteinDAGTheorem}, this can be done in time $\bigO(|G|)$. In the following, for every node $q$ of $\mathcal{H}(G)$, we denote by $P_q$ the $s$-to-$t$-path of $G$ that is represented by $q$. We use an \emph{auxiliary heap}, which can store nodes of $\mathcal{H}(G)$ and which is a min-heap with respect to the weights of the $s$-to-$t$-paths represented by the nodes of $\mathcal{H}(G)$ it stores. Initially, we store the root of $\mathcal{H}(G)$ in the auxiliary heap. This concludes the preprocessing phase.\par
In the enumeration phase, we repeat the following step until the auxiliary heap is empty. We remove the minimum element of the auxiliary heap, which represents some node $q$ of $\mathcal{H}(G)$, and then we insert all of $q$'s children of $\mathcal{H}(G)$ into the auxiliary heap. After that, we produce the output element $\labelfunc{P_q}$.\par
This clearly produces an enumeration of the labels of the $s$-to-$t$-paths of $G$. Moreover, this enumeration is a ranked enumeration, due to the fact that $\mathcal{H}(G)$ is a min-heap with respect to the weights of the stored $s$-to-$t$-paths.\par
The preprocessing time is $\bigO(|G|)$. Every step of the enumeration phase requires one remove-min-operation with respect to the auxiliary heap, the traversals of at most $4$ edges in $\mathcal{H}(G)$ and the output of a label. Thus, the time needed for one step is $\bigO(|s| + \log h)$, where $s$ is the output label and $h$ is the current size of the auxiliary heap, which is bounded by the number of elements that have been produced as output (including the element produced in this step). Hence, if the enumeration phase produces the elements $s_1, s_2, \ldots, s_k$, then each element $s_i$ with $i \in [k]$ has a delay of $\bigO(|s_i| + \log i)$. Next, we explain how to tweak this approach to ensure that $h \leq|G|$, making the delay $\bigO(|s_{i}| + \min\{\log i, \log |G|\})$.

Instead of storing all the elements on the auxiliary heap, we will maintain them on multiple linked list. Let $D_{G}$ be the Eppstein-DAG that represents the heap $\mathcal{H}(G)$ of $G$'s $s$-to-$t$-paths, i.\,e., the nodes of $\mathcal{H}(G)$ correspond to paths in $D_{G}$ (see Appendix~\ref{sec:AppendixEppstein}).
For every edge $e$ of $D_{G}$, we maintain a separate linked list $L(e)$. Each element of $L(e)$ is a node of $\mathcal{H}(G)$ corresponding to a path of $D_{G}$ ending
with $e$. We will ensure that the elements are arranged in the non-decreasing order of their weights.
This allows us to store only the first element of each $L(e)$ (for non-empty $L(e)$) on the auxiliary heap, and make sure that the size of the auxiliary
heap is at most $|G|$.
Consider a node $q$ of $\mathcal{H}_{G}$ that has been output by the procedure. Then, the weight of every child $p$ of $q$ is equal to that of $q$
increased by the weight of the corresponding edge $e$ of $D_{G}$. We append $p$ to $L(e)$, and observe that this maintains the invariant: the previous element $p'$ of
$L(e)$ has been obtained as a child of some $q'$ that has been output earlier by the procedure, thus the weight of $q'$ is at most that of $q$, and
and the weight of $p'$ is the weight of $q'$ increased by the weight of $e$.
\end{proof}

\section{Improved Enumeration Algorithms}

We now obtain a ranked enumeration algorithm for $\enumProbTransducer^f$ with the reduction described in Section~\ref{sec:rankedInformationExtraction} and Corollary~\ref{simpleEnumGraphCorollary} (the preprocessing time follows by Observation~\ref{sizeOfDAGObservation}). 
As explained in the introduction, this constitutes a substantial improvement over Theorem~\ref{BorhisResultTheorem}.

\begin{theorem}\label{simpleEnumTransducerTheorem}
There is a ranked enumeration algorithm for $\enumProbTransducer^f$ such that, on input $\mathcal{T}$ and $\doc$, the preprocessing time is $\bigO(|\mathcal{T}| |\doc|)$ and, if $(s_1, s_2, \ldots, s_k)$ is the output of the algorithm, then the delay of $s_i$ is $\bigO(|s_i| + \min\{ \log i, \log(|\mathcal{T}| |\doc|)\})$ for every $i \in [k]$. 
\end{theorem}

We will now further improve the delay. We first observe that the weights of all paths in the considered DAG admit a special structure.
Let $g_{1},g_{2},\ldots,g_{t} \in \group$ and let $n \in \mathbb{N}$. Any expression $\sum^t_{i = 1} \alpha_{i} g_{i}$, where $\alpha_{1},\alpha_{2},\ldots,\alpha_{t}\in [0,n]$ and $\sum^t_{i = 1} \alpha_{i} \leq n$ is called an \emph{$n$-sum} (\emph{over $g_1, g_2, \ldots, g_t$}). We will assume that every $n$-sum $\sum^t_{i = 1} \alpha_{i} g_{i}$ is represented by the tuple $(\alpha_{1},\alpha_{2},\ldots,\alpha_{t})$.
After a simple preprocessing, we are able to evaluate such a sum in $\bigo(t)$ time on demand.

\begin{observation}
\label{obs:compare}
In $\bigo(n t)$ time, we can compute all values $k g_i$ for every $k \in [n]$ and $i \in [t]$. This means that within an $\bigo(n t)$ preprocessing, we can compute a data structure that allows us to compute for a given tuple $(\alpha_{1},\alpha_{2},\ldots,\alpha_{t})$ the group-element $\sum^t_{i = 1} \alpha_{i} g_{i}$ in time $\bigo(t)$. 
\end{observation}

Our main technical result for improving the delay concerns the problem of sorting given $n$-sums.
The proof of this lemma is deferred to Section~\ref{sec:rounding}. 

\begin{lemma}\label{mainSortingLemma}
There is an algorithm that sorts any set $X$ of $n$-sums over some $g_{1},g_{2},\ldots,g_{t} \in \group$ with $|X| = N\geq n$
in time $\Oh(\poly(t)N)$ with high probability. Formally, for any $c \in \mathbb{N}$, the output of the algorithm is always correct,
and the bound on the time holds with probability at least $1-1/N^{c}$.
\end{lemma}

In the main results of this paper, stated as Theorem \ref{PathEnumWithSortingTheorem} and its consequence Theorem \ref{advancedEnumTheorem}, we will use the result stated in Lemma~\ref{mainSortingLemma}, namely an efficient method of sorting sets of $n$-sums over the elements of $\group$, together with a further well-known data structures result, stated below as Lemma \ref{fredericksonResult}, for improving the delay of our enumeration algorithm for $\enumProbGraph^f$ (and therefore $\enumProbTransducer^f$). However, let us briefly mention that some natural questions arising in this context are why is such a novel sorting algorithm needed and why are standard approaches to sorting not sufficient in this context? To not break the flow of this section, we will address these question at its end, in Remark \ref{sortingComparison}, after showing how our sorting method is actually used.

As already mentioned above, alongside Lemma \ref{mainSortingLemma}, we need the following well-known result.

\begin{lemma}[Frederickson~\cite{Frederickson93}]\label{fredericksonResult}
Given some $k \in \mathbb{N}$ and a min-heap $H$ that allows us to move from a node to its children in constant time, we can find the element of rank $k$ of $H$ in time $\bigO(k)$.
\end{lemma}

Observe that Lemma~\ref{fredericksonResult} also means that we can get in $\bigo(k)$ time all nodes of $H$ with rank $1, 2, \ldots, k$, as follows: Compute first the element of rank $k$, which has some key $\ell$, and then explore $H$ starting from the root and output all elements with a key not larger than $\ell$ (or, in the case that there are more than one element with key $\ell$, stop as soon as $k$ elements have been produced).

For every directed graph $G = (V, E, \weightName, \labelName)$ with edge weights and edge labels, let $\weightnumber{G} = |\{\weightfunc{e} \mid e \in E\}|$ denote the number of distinct weights that occur in $G$. We can show the following theorem.

\begin{theorem}\label{PathEnumWithSortingTheorem}
There is a ranked enumeration algorithm for $\enumProbGraph^f$ such that, on any input $(G, s, t)$, the preprocessing time is $\bigO(|G| + \poly(\weightnumber{G})|V|)$ and every output element $s$ has a delay of $\bigo(|s| \poly(\weightnumber{G}))$. The guarantees on the preprocessing time and on the delay hold with high probability in the size of the input and the size of the output generated so far.
\end{theorem}

\begin{proof}
In the preprocessing phase, we construct a data structure that represents a heap representation $\mathcal{H}(G)$ with degree $4$ of the $s$-to-$t$-paths of $G$, such that we can move from a node $q$ of $\mathcal{H}(G)$ to any of its children in constant time, and, for every node $q$ of $\mathcal{H}(G)$ that represents an $s$-to-$t$-path $P_q$, we can retrieve $\weightfunc{P_q}$ in constant time and $\labelfunc{P_q}$ in time $\bigo(|\labelfunc{P_q}|)$. According to Theorem~\ref{EppsteinDAGTheorem}, this can be done in time $\bigO(|G|)$.

We define $n = |V|$ and assume that $g_1, g_2, \ldots, g_t$ are the weights that occur in $G$. We observe that since each edge of $G$ has a weight from $g_1, g_2, \ldots, g_t$ and every $s$-to-$t$-path has at most $n-1$ edges (since $G$ is a DAG), the weight of any $s$-to-$t$-path is an $n$-sum over $g_1, g_2, \ldots, g_t$. According to Lemma~\ref{mainSortingLemma}, we can sort any $N \geq n$ such $n$-sums over $g_1, g_2, \ldots, g_t$ in time $\Oh(\poly(t)N)$ with high probability in $N$.
Still in the preprocessing, we use Lemma~\ref{fredericksonResult} to extract the smallest $n$ elements of $\mathcal{H}(G)$ in time $\bigo(n)$, we sort them in $\Oh(\poly(t)n)$ time with high probability in $n$, using the algorithm of Lemma~\ref{mainSortingLemma}, and store them in the sorted order. This concludes the preprocessing, which clearly can be done in time $\bigO(|G| + \poly(t)n)$ with high probability.\par

The enumeration phase proceeds in several epochs. For every $i = 1, 2, 3, \ldots$, the $i^{\text{th}}$ epoch will enumerate the labels of the $2^{i-1} n$ elements of $\mathcal{H}(G)$ with ranks in $(2^{i-2} n, 2^{i-1} n]$ (or in $[1, n]$ for $i = 1$). As an invariant, we assume that when the $i^{\text{th}}$ epoch starts, we already have at our disposal in sorted order all elements whose labels are to be enumerated in this epoch. This invariant holds for the first epoch, due to the preprocessing mentioned above. 

In epoch $i$, we do the following: we extract the smallest $2^{i} n$ elements of $\mathcal{H}(G)$ with Frederickson's algorithm, sort them with Lemma~\ref{mainSortingLemma}, and store in the sorted order all elements with ranks from $(2^{i-1} n, 2^{i} n]$ (i.\,e., the elements to be enumerated in the next epoch $i+1$). Due Lemmas~\ref{fredericksonResult}~and~\ref{mainSortingLemma}, we can assume that the number of computational steps of this procedure is bounded by $c f(t) 2^{i} n$ with high probability in $2^{i}n$ for some constant $c$ and polynomial $f$ (that do not depend on $i$). At the beginning of epoch $i$, we initialise a counter with $0$ that simply counts the computational steps. Whenever the counter reaches value $\frac{c f(t) 2^{i} n}{2^{i-1} n}$, we output in time $\bigo(|L|)$ the label $L$ that corresponds to the next element of the list of pre-computed elements from $\mathcal{H}(G)$ and we set the counter back to $0$. Since we need at most $c f(t) 2^{i} n$ computational steps in total, and we have $2^{i-1} n$ pre-computed elements in epoch $i$, we will never run out of elements to output.

It can be easily verified that in this way epoch $i$ does in fact enumerate the labels of the $2^{i-1} n$ elements of $\mathcal{H}(G)$ with ranks in $(2^{i-2} n, 2^{i-1} n]$ (or in $[1, n]$ for $i = 1$). Moreover, the delay of any output element $L$ is clearly bounded by $|L| \frac{c f(t) 2^{i} n}{2^{i-1} n} \leq c' f(t) |L|$ with high probability in the size of the input and the size of the output generated so far, for a constant $c'$.
\end{proof}

We can directly use Theorem~\ref{PathEnumWithSortingTheorem} in order to get a ranked enumeration algorithm for $\enumProbTransducer^f$. Let $\doc$ be a document and let $\mathcal{T}$ be a cost transducer.
We compute the DAG representation $G_{\doc, \mathcal{T}}$ in time $\bigo(|\mathcal{T}||\doc|)$, and we observe that $G_{\doc, \mathcal{T}}$ is a DAG with $|\mathcal{T}||\doc|$ nodes and $\weightnumber{G_{\doc, \mathcal{T}}} = \bigo(|\mathcal{T}|)$. Hence, Theorem~\ref{PathEnumWithSortingTheorem} yields the following theorem, which is our main result.

\begin{theorem}\label{advancedEnumTheorem}
There is a ranked enumeration algorithm for $\enumProbTransducer^f$ such that, on any input $(\mathcal{T}, \doc)$, the preprocessing time is $\bigO(\poly(|\mathcal{T}|)|\doc|)$ and every output element $s$ has a delay of $\bigo(|s| \poly(|\mathcal{T}|))$. The guarantees on the preprocessing time and on the delay hold with high probability in the size of the input and the size of the output generated so far.
\end{theorem}

To completely substantiate our results, we still need to prove Lemma \ref{mainSortingLemma} and, as already announced above, comment on why its result cannot be achieved via some standard approaches to sorting. We will address the latter point below, while the proof of Lemma \ref{mainSortingLemma} is the focus of the next section. 
\begin{remark}\label{sortingComparison} 
Recall that Lemma \ref{mainSortingLemma} states that one can sort any set $X$ of $n$-sums over some $g_{1},g_{2},\ldots,g_{t} \in \group$ with $|X| = N\geq n$
in time $\Oh(\poly(t)N)$ with high probability. To begin with, it is worth noting that, in our setting, $N$ is upper bounded by a polynomial in $n$. Hence, if we restrict the setting of our approach to the case when $\group$ is the set of integers and each weight $g_1,\ldots, g_t$ can be represented in $\Oh(\log n)$ bits (and, as such, every $n$-sum contained in $X$ can also be represented in $O(\log n)$ bits), one can use radix sort \cite{Knuth73} to order the elements of the set $X$ deterministically in linear time. So, in that particular case, computing the value of the elements of $X$ and sorting them can be done deterministically in $\Oh(tN)$ time overall. In conclusion, in the restricted setting, this standard approach can be used instead of Lemma \ref{mainSortingLemma} in the algorithm of Theorem \ref{PathEnumWithSortingTheorem}. Already in a slightly more general setting, where we only assume that the weights are integers which fit in a memory word (of size $\omega \geq \log n$), but we do not generally assume that they consists of only $\Oh(\log n)$ bits, this does not work anymore. In this case, each $n$-sum still fits in a constant number of memory words, but it is not known whether we can sort the respective $N$ $n$-sums in $O(\poly(t)N)$ time (and, in general, any $N\geq n$ integers, about which we only know that they fit in a single memory word of size $\omega\geq \log n$ each, in $\Oh(N)$ time). For a discussion of sorting methods which could be used in this setting, see, e.g., \cite{AnderssonHNR98,HanT02,AnderssonT07} and the references therein; our result is already useful here, as it exhibits an algorithm sorting the elements of the set $X$ in $\Oh(\poly(t)N)$, albeit only with high probability. Finally, the setting in which we develop our algorithm (as defined in \cite{BourhisEtAl2021}) is more general than all those mentioned above, as we do not make any assumption on the structure and elements of $\group$, except that each of these elements fits in one memory-word. None of the aforementioned sorting algorithms seems to work in this framework, so, we needed to develop the novel sorting algorithm from Lemma~\ref{mainSortingLemma}, which specifically addresses sets of $n$-sums of $\group$-elements, and rely on it in the proof of Theorem \ref{PathEnumWithSortingTheorem}. 
\end{remark}

\section{Sorting $n$-Sums in Linear Time With High Probability}\label{sec:rounding}

This section is devoted to proving Lemma~\ref{mainSortingLemma}: There is an algorithm that sorts any set $X = \{x_{1},x_{2},\ldots,x_{N}\}$ of $n$-sums over some $g_{1},g_{2},\ldots,g_{t} \in \group$ with $|X| = N\geq n$ in time $\Oh(\poly(t)N)$ with high probability in $N$. We structure this proof on several paragraphs for readability purposes.

\paragraph{$ \S $ Algebraic preliminaries.} For basic results and notations regarding linear algebra and, respectively, probability theory, we refer to the standard handbooks \cite{lang} and, respectively, \cite{mitzenmacher}. 

Let us begin with a series of preliminary results. Let $x_{i}=\sum_{j=1}^{t}\alpha^{i}_{j}g_{j}$ for every $i \in [N]$. Then, $(\alpha^{i}_{1},\alpha^{i}_{2},\ldots,\alpha^{i}_{t})=(\alpha^{i'}_{1},\alpha^{i'}_{2},\ldots,\alpha^{i'}_{t})$ implies $x_{i}=x_{i'}$, so we can identify identical vectors $(\alpha^{i}_{1},\alpha^{i}_{2},\ldots,\alpha^{i}_{t})$ in $\Oh(t(N+n))=\Oh(tN)$
time by radix sorting and remove elements corresponding to duplicated vectors from $X$. After such preliminary
reduction, $N\leq n^{t}$, so $\log N\leq t\log n$. Next, it will be convenient to assume that the actual $n$-sums
(and not just their vectors) are pairwise distinct, which can be assumed without losing
generality by the following simple reasoning.

\begin{proposition} \label{prop:distinct} 
Let $n \in \mathbb{N}$. Sorting any given $n$-sums over $g_{1},g_{2},\ldots,g_{t} \in \group$ can be reduced in linear time to sorting $N$ \emph{distinct} $(N+n)$-sums over $g'_{1},g'_{2},\ldots,g'_{t}$, where $g'_{1},g'_{2},\ldots,g'_{t}$ are elements of a new ordered group $(\group', \plusgg, \zerogg, \leqgg)$ with an effective representation.
\end{proposition}

\begin{proof}
The elements of $\group'$ are pairs $(g,x)$, where $g\in \group$ and $x\in \Z$, with $\plusgg$ being the coordinate-wise addition
and $\leqgg$ the lexicographical comparison. It is clear that $\group'$ admits an effective representation.
Next, we define $g'_{i} = (g_{i},0)$, for $i=1,2,\ldots,t$, and $g'_{t+1}= (0,1)$. Now, let the $k$-th given $n$-sum
be $\sum_{i=1}^{t} \alpha_{i}g_{i}$. We convert it to $(\sum_{i=1}^{t}\alpha_{i}g'_{i}) + k \cdot g'_{t+1}$.
This ensures that the obtained $(n+N)$-sums are pairwise distinct, as each of them has distinct second coordinate.
Arranging them in the strictly increasing order gives us the original $n$-sums in the non-decreasing order.
\end{proof}

To avoid clutter, from now on we will assume that the goal is to sort $N$ given $(N+n)$-sums over $g_{1},g_{2},\ldots,g_{t}\in\group$
(instead of talking about $g'_{1},g'_{2},\ldots,g'_{t} \in \group'$).

Further, we need some background on ordered abelian groups.
Hahn~\cite{Hahn1907} showed that
any such group has a certain nice structure (also see a short proof by Clifford~\cite{Clifford1954} that extends
a previous proof by Hausner and Wendel~\cite{HausnerW1952} for the case of ordered vector spaces over the reals).
To precisely describe this structure, we need a few definitions. Two elements $x,y\in \group$ are Archimedean equivalent if there
exist natural numbers $m,n$ such that $|x| \leqg n |y|$ and $|y| \leqg m|x|$. This naturally defines the Archimedean
equivalence classes of $\group$, denoted $\Omega$, and we let $[x]$ denote the Archimedean equivalence class of $x\in \group$. We
define a total order over $\Omega$, denoted $\lessa$, as follows:
for two classes $[x], [y] \in \Omega$, we have $[x] \lessa [y]$ when $m |y| \lessg |x|$ for every natural number $m$
(equivalently, $m y \lessg |x|$ for every integer number $m$).
$\R^{\Omega}$ denotes the set of all functions from $\Omega$ to $\R$ with a well-ordered support (where the support, denoted $\supp(\cdot)$, is the subset of $\Omega$ which is not mapped to $0$ by the respective function),
which allows us to define the lexicographical ordering on $\R^{\Omega}$ in a natural way by saying that, for
$f,g\in \R^{\Omega}$, we have $f\leqlex g$ when $f(x)<g(x)$ for $x=\min \{  y\in \supp(f)\cap \supp(g) | f(y) \neq g(y) \}$,
which is always defined by the assumption of $\supp(f)$ and $\supp(g)$ being well-ordered. 
This gives us a natural way to define an ordered abelian group over $\R^{\Omega}$, where the operation is the addition of functions.
For two ordered abelian groups $(\group, \plusg, \zerog, \leqg)$ and $(\group', \plusg', \zerog', \leqg')$,
an order-preserving embedding is a function $h : \group \rightarrow \group'$ such that, for all $x,y\in \group$, we have:
\begin{enumerate}
\item $h(x \plusg y) = h(x) \plusg' h(y)$ (group homomorphism),
\item $x \leqg y$ if and only if $h(x) \leqg' h(y)$ (order embedding). 
\end{enumerate}

\begin{theorem}[Hahn's embedding theorem~\cite{Hahn1907,Clifford1954}]
\label{thm:hahn}
$(\group, \plusg, \zerog, \leqg)$ has an order-preserving embedding into $\R^{\Omega}$, where $\R^{\Omega}$ has the
lexicographical ordering.
\end{theorem}

As we can only access $\group$ by calculating expressions of the form $\sum_{i=1}^{t} \alpha_{i} g_{i}$,
we can assume that $\group$ is finitely generated by $\{g_{1},g_{2},\ldots,g_{t}\}$; to avoid clutter, we will not write
anymore that the sums are for $i$ from $1$ to $t$, and consider this implicit. In this case,
Hahn's embedding theorem implies the following.

\begin{lemma}
\label{lem:embedding}
If $(\group, \plusg, \zerog, \leqg)$ is finitely generated by $\{g_{1},g_{2},\ldots,g_{t}\}$ then there exist $\vec{v_{1}},\vec{v_{2}},\ldots,\vec{v_{t}}\in \R^{t}$
such that for any $\alpha_{1},\alpha_{2},\ldots,\alpha_{t}\in \Z$ we have $\sum_{i}\alpha_{i}g_{i} \leqg 0$ if and only if
$\sum_{i}\alpha_{i}\vec{v_{i}} \leqlex \vec{0}$.
\end{lemma}

\begin{proof}
Any element of $\group$ is of the form $\sum_{i} \gamma_{i} g_{i}$, where $\gamma_{1},\gamma_{2},\ldots,\gamma_{t}\in \Z$,
or $\langle \vec{\gamma}, \vec{g} \rangle$ in short with $\vec{\gamma} = (\gamma_{1},\gamma_{2},\ldots,\gamma_{t})$ and $\vec{g} = (g_{1}, g_{2}, \ldots, g_{t})$.
We claim that there are at most $t$ Archimedean equivalence classes of $\group$.
Assume otherwise, and choose $t+1$ such classes $[x_{1}], [x_{2}], \ldots, [x_{t+1}]$, where
$[x_{t+1}] \lessa [x_{1}], [x_{2}], \ldots, [x_{t}]$ and without losing the generality $0_{\group} \lessg x_{t+1}$.
Write $x_{i} = \langle \vec{\alpha_{i}} ,\vec{g} \rangle$ with $\vec{\alpha_{i}}\in \Z^{t}$, for $i=1,2,\ldots,t+1$.
Then, the $t+1$ vectors $\vec{\alpha_{i}}$ cannot be linearly independent over the rationals, so
(by multiplying by the least common denominator) there are integer coefficients
$\beta_{1},\beta_{2},\ldots,\beta_{t}$ such that $\sum_{i=1}^{t}\beta_{i} \vec{\alpha_{i}} = \vec{\alpha_{t+1}}$. 
This implies $\sum_{i=1}^{t} \beta_{i} \vec{\alpha_{i}} \vec{g}  = \vec{\alpha_{t+1}} \vec{g}$, and, since $x_{i} = \langle \vec{\alpha_{i}} ,\vec{g} \rangle$ for all $i=1,2,\ldots,t+1$, we obtain $\sum_{i=1}^{t} \beta_{i} x_i = x_{t+1}$ and therefore it follows that $\sum_{i=1}^{t}t\beta_{i} x_{i}= t x_{t+1}$.
Next, from $[x_{t+1}] \lessa [x_{i}]$ we have
$t\beta_{i}x_{i} \lessg |x_{t+1}|$, for $i=1,2,\ldots,t+1$. Adding the inequalities we conclude that
$t x_{t+1}=\sum_{i=1}^{t}t\beta_{i} x_{i} \lessg t |x_{t+1}|$, which is a contradiction with $\zerog \lessg x_{t+1}$.
Thus indeed $\Omega$ is a finite set of cardinality at most $t' \leq t$ so there are at most $t' \leq t$ Archimedean equivalence classes of $\group$ as claimed above.
Then $\R^{\Omega}$ consists of length-$t'$ vectors composed of real numbers.
We can pad the vectors with zeroes to make their lengths exactly $t$.
By \Cref{thm:hahn}, there exists a function $h: \group \rightarrow \R^{\Omega}$ such that,
for all $x,y\in \group$, $h(x \plusg y) = h(x) + h(y)$ and $x \leqg y$ if and only if $h(x) \leqlex h(y)$.
We take $\vec{v_{i}} = h(g_{i})$, for every $i=1,2,\ldots,t$.
Then, consider any $\alpha_{1},\alpha_{2},\ldots,\alpha_{t}\in \Z$. By the properties of $h$,
$h(\sum_{i}\alpha_{i}g_{i})=\sum_{i}\alpha_{i}h(g_{i})=\sum_{i}\alpha_{i}\vec{v_{i}}$, and
$\sum_{i}\alpha_{i}g_{i} \leqg \zerog$ if and only if $\sum_{i} \alpha_{i} \vec{v_{i}}= h(\sum_{i}\alpha_{i}g_{i}) \leqlex \vec{0}$.
\end{proof}

Further, we will later verify that our algorithm only accesses $\group$ by performing comparisons of the form
$\sum_{i} \alpha_{i}g_{i} \leqg \zerog$, where $\alpha_{1},\alpha_{2},\ldots,\alpha_{t}\in [-2(N+n),2(N+n)]$.
In such a case, we can actually think that the elements of $\group$ are real numbers (in fact, to this end, it is enough
that every $\alpha_{i}$ comes from a finite set).
We stress that these real numbers are never computed by the algorithm, and this is just
a formal proof that, without losing generality, we can think that $\group$ is a subgroup of $(\R, +, 0, \leq)$,
which is a crucial insight towards understanding both the correctness and complexity of this algorithm.
The following lemma follows from \Cref{lem:embedding}.

\begin{lemma}
\label{lem:real}
There exist $\hat{g}_{1}, \hat{g}_{2}, \ldots, \hat{g}_{t} \in \R$ such that, for every $\alpha_{1},\alpha_{2},\ldots,\alpha_{t} \in [-2(N+n),2(N+n)]$, we have $\sum_{i} \alpha_{i}g_{i} \leqg \zerog$ if and only if $\sum_{i} \alpha_{i}\hat{g}_{i} \leq 0$.
\end{lemma}

\begin{proof}
By \Cref{lem:embedding}, there are $\vec{v_{1}},\vec{v_{2}},\ldots,\vec{v_{t}}\in \R^{t}$ such that,
for any $\alpha_{1},\alpha_{2},\ldots,\alpha_{t}\in [-2(N+n),2(N+n)]$,
we have $\sum_{i}\alpha_{i}g_{i} \leqg 0$ if and only if $\sum_{i}\alpha_{i}\vec{v_{i}} \leqlex \vec{0}$.
Now consider the following mapping:
\[ f(y_{1},y_{2},\ldots,y_{t}) = \sum_{i} y_{i} \epsilon^{i} \]
and consider applying it on all points of the form $\sum_{i}\alpha_{i}\vec{v_{i}}$, for $\alpha_{1},\alpha_{2},\ldots,\alpha_{t} \in [-2(N+n),2(N+n)]$.
We claim that, for sufficiently small $\epsilon > 0$, for every $\alpha_{1},\alpha_{2},\ldots,\alpha_{t} \in [-2(N+n),2(N+n)]$
we have $\sum_{i}\alpha_{i}\vec{v_{i}} \leqlex \vec{0}$ if and only if $f(\sum_{i}\alpha_{i} \vec{v_{i}}) \leq 0$.
Indeed, it is easy to see that for any point $\vec{y}\in \R^{t}$ and sufficiently small $\epsilon > 0$ we have
$\vec{y} \leqlex \vec{0}$ if and only $f(\vec{y}) \leq 0$, and then the claim holds because we have finitely many points.
Next, $f$ is a multi-linear mapping, thus for any $\alpha_{1},\alpha_{2},\ldots,\alpha_{t} \in [-2(N+n),2(N+n)]$
we have $\sum_{i}\alpha_{i}\vec{v_{i}} \leqlex \vec{0}$ if and only if $\sum_{i}\alpha_{i}f(\vec{v_{i}}) \leq 0$.
We set $\hat{g}_{i} = f(\vec{v_{i}})$ and conclude that, for every $\alpha_{1},\alpha_{2},\ldots,\alpha_{t}\in [-2(N+n),2(N+n)]$,
$\sum_{i}\alpha_{i}g_{i} \leqg \zerog$ if and only if $\sum_{i} \alpha_{i}\hat{g}_{i} \leqg 0$.
\end{proof}

\paragraph{$\S $ The point-location problem and its relation to our problem.}

Before presenting the actual algorithm showing Lemma \ref{mainSortingLemma}, we also need to recap the approach of Kane, Lovett, and Moran~\cite{KaneLM19}, and see its connections to our setting.
For a finite set $H = \{h_1, h_2, \ldots, h_{|H|}\} \subseteq \R^{t}$, let the function $\arrange_{H} : \R^{t} \to \{-,0,+\}^{|H|}$ be defined as
$\arrange_{H}(x) := (\sign(\langle x, h_1 \rangle), \sign(\langle x, h_2 \rangle), \ldots, \sign(\langle x, h_{|H|} \rangle))$
(where $\langle \cdot, \cdot \rangle$ denotes the inner product on the vector space $\R^{t}$, and $\sign$ is the sign-function on $\R$).

Kane, Lovett, and Moran~\cite{KaneLM19} considered the problem of determining $\arrange_{H}(x)$ for any given $x\in \R^{t}$ with few queries of the form $\sign(\langle x,h\rangle)$, for $h\in H$ (label queries)
and $\sign(\langle x, h'-h''\rangle)$, for $h',h''\in H$ (comparison queries); this is \emph{the point-location in an hyperplane-arrangement problem}.

Before proceeding with further definitions, let us comment on how this could be possibly useful in our problem.
Recall that our current goal (according to Proposition \ref{prop:distinct} and the comment made after it) is to sort a set $X$ of $N \geq n$ given $(N+n)$-sums, where the $i$-th given $(N+n)$-sum is $x_i = \sum_{j = 1}^{t}\alpha^{i}_{j} g_{j}$. First, let us consider how many expressions of the form $\sum_{j}\alpha_{j} g_{j} \leqg \zerog$
need to be evaluated in order to uniquely determine the sorting permutation $\pi$ of $X$,
and disregard the computational complexity of actually finding $\pi$ given the results of all the comparisons.
It is clear that it is necessary and sufficient to resolve $\sum_{j} \alpha^{i}_{j} g_{j} \leqg \sum_{j} \alpha^{i'}_{j} g_{j}$,
for every $i,i'$, or in other words $\sum_{j} (\alpha^{i}_{j} - \alpha^{i'}_{j}) g_{j} \leqg \zerog$.
From the point of view of evaluating such expression, \Cref{lem:real} tells us that we can think that
$g_{1},g_{2},\ldots,g_{t}\in \R$.
Thus, we can reformulate the problem by defining
$H = \{ (\alpha^{i}_{1}-\alpha^{i'}_{1}, \alpha^{i}_{2}-\alpha^{i'}_{2},\ldots,\alpha^{i}_{t}-\alpha^{i'}_{t}) : 1\leq i, i' \leq N\}$
and saying that we want to determine $\arrange_{H}(g)$, where $g=(g_{1},g_{2},\ldots,g_{t})$.
This will be done with few label and comparison queries. We note that a label query evaluates
an expression $\sum_{j}\beta_{j}g_{j} \leqg \zerog$, where $\beta_{1},\beta_{2},\ldots,\beta_{t}\in [-(N+n),N+n]$,
while a comparison query evaluates an expression 
$\sum_{j}\beta_{j}g_{j} \leqg \zerog$, where $\beta_{1},\beta_{2},\ldots,\beta_{t}\in [-2(N+n),2(N+n)]$,
thus we can still think that $g_{1},g_{2},\ldots,g_{t} \in \R$ (and this is, indeed, the reason for having $2(N+n)$ in the statement of \Cref{lem:real}). The main result of Kane, Lovett, and Moran~\cite{KaneLM19} shows that it is possible to determine
$\arrange_{H}(g)$ by evaluating only $O(t\log(t(N+n))\log N)$ such expressions corresponding to label and comparison queries. 
However, this is an existential proof, and does not tell us how to actually find which expressions to evaluate, and even if we knew that, it would still be quite problematic, given that we have $N^{2}$ possible expressions, and we cannot explicitly operate on
all of them. Thus, we will need a more refined approach that can be implemented efficiently. This requires
looking inside the approach of Kane, Lovett, and Moran~\cite{KaneLM19}, and some more definitions.

For $S \subseteq H$, and $h, x\in \R^{t}$, we say that $S$ infers $h$ at $x$ if $\sign(\langle h,x\rangle)$
is fully determined by the answers to the label and comparison queries on $S$.
We define $\infer(S,x)$ as the set of all $h\in \R^{t}$ such that $S$ infers $h$ at $x$. 
The inference dimension of $H\subset \R^{t}$ is defined as the smallest $d\geq 1$ such that, for any $S\subset H$
of size at least $d$, and for any $x\in \R^{t}$, there exists $h\in S$ such that $S\setminus \{h\}$ infers $h$ at $x$.
While it might be difficult to precisely determine the inference dimension of a specific $H$, it can be upper bounded
as follows. For $h\in \R^{t}$, where $h=(h_{1},h_{2},\ldots,h_{t})$, we define $||h||_{1}=\sum_{i=1}^{t} |h_{i}|$.
Then, if every $h\in H$ has integer coefficients such that $||h||_{1} \leq w$ then the inference dimension of $H$ is $d=\Oh(t\log w)$.
This follows from the following lemma that upper bounds the inference dimension of the set consisting of all $h\in \Z^{t}$
such that $||h||_{1}\leq w$, and the fact that the inference dimension of a smaller set cannot be larger.

\begin{lemma}[{\cite[Theorem 1.10]{KaneLM19}}]
\label{lem:inference}
The inference dimension of $\{ h\in \Z^{t} : ||h||_{1} \leq w\}$ is $d=\Oh(t\log w)$.
\end{lemma}

If we are able to bound the inference dimension of $H$ (for example, by the above lemma), then we have the following
lemma.

\begin{lemma}[{Corollary of~\cite[Lemma 4.2]{KaneLM19}}]\label{lem:whp}
Choose any $c \in \mathbb{N}$.
Let $H\subseteq \R^{t}$ be a finite set with inference dimension at most $d$, and $S\subseteq H$ be its uniformly chosen subset of size $s \cdot c \log n$. Then, for every $x\in \R^{t}$:
\[
\Pr_{S}\left[|H \setminus \infer(S,x)| \geq \frac{2d}{s+1}|H|\right] \leq \frac{1}{n^{c}}.
\] 
\end{lemma}

\begin{proof}
We need the following lemma. It was originally stated with $s=2d$, but a direct modification of the proof gives the following
more general bound.

\begin{lemma}[{\cite[Lemma 4.2]{KaneLM19}}]
\label{lem:shrink}
Let $H\subseteq \R^{t}$ be a finite set with inference dimension at most $d$, and $S\subseteq H$ be its uniformly chosen
subset of size $s$. Then, for every $x\in \R^{t}$:
\[
\mathbb{E}_{S}[|\infer(S,x) \cap H|] \geq \frac{s+1-d}{s+1}|H|.
\] 
\end{lemma}

Now, consider the following random process. We first uniformly choose a subset $S\subseteq H$ of size $s\cdot c\log n$, and then
uniformly choose size-$s$ subsets $S_{1},S_{2},\ldots,S_{c\log n} \subseteq S$. 
We observe that $\infer(S_{i},x) \subseteq \infer(S,x)$ when $S_{i}\subseteq S$, so:
\begin{align*}
\Pr_{S}\left[|H \setminus \infer(S,x)| \geq \frac{2d}{s+1}|H|\right] = & \Pr_{S,S_{1},S_{2},\ldots,S_{c\log n}}\left[|H \setminus \infer(S,x)| \geq \frac{2d}{s+1}|H|\right]  \\
\leq & \Pr_{S,S_{1},S_{2},\ldots,S_{c\log n}}\left[\forall_{i=1,2,\ldots,c\log n}|H \setminus \infer(S_{i},x)| \geq \frac{2d}{s+1}|H| \right].
\end{align*}
Next, we observe that the obtained distribution on the subsets $S_{1},S_{2},\ldots,S_{c\log n}$ is exactly the same as if
$S_{1},S_{2},\ldots,S_{c\log n}\subseteq H$ were uniformly chosen subsets of size $s$.
Thus:
\[
\Pr_{S}\left[|H \setminus \infer(S,x)| \geq \frac{2d}{s+1}|H|\right] \leq \Pr_{S_{1},S_{2},\ldots,S_{c\log n}}\left[\forall_{i=1,2,\ldots,c\log n}|H \setminus \infer(S_{i},x)| \geq \frac{2d}{s+1}|H| \right].
\]
If $S_{i}\subseteq H$ is a uniformly chosen subset of size $s$ then, by \Cref{lem:shrink} and Markov's inequality, we conclude:
\begin{equation*}
\Pr_{S_{i}}\left[| H \setminus \infer(S_{i},x) | \geq \frac{2d}{s+1} |H| \right] \leq \frac{\mathbb{E}_{S}[|H \setminus \infer(S,x)|]}{\frac{2d}{s+1} |H|} \leq \frac{(1 - \frac{s+1-d}{s+1})|H|}{\frac{2d}{s+1} |H|} = \frac{1}{2}\,.
\end{equation*}
The first inequality follows from Markov's inequality, while the second inequality uses the fact that $\mathbb{E}_{S}[|H \setminus \infer(S,x)|] \leq (1- \frac{s+1-d}{s+1})|H|$, which is an immediate result from \Cref{lem:shrink}.

The events $|H \setminus \infer(S_{i},x)| \geq \frac{2d}{s+1}|H|$ are independent, so:
\[
\Pr_{S}\left[|H \setminus \infer(S,x)| \geq \frac{2d}{s+1}|H|\right] \leq \left(\frac{1}{2}\right)^{c\log n} = \frac{1}{n^{c}}
\]
as required.
\end{proof}

\paragraph{$\S $ High-level overview of the algorithm.} After presenting the main technical ingredients of our algorithm, we now provide a high-level description of the algorithm. 

Recall that we want to sort a set $X = \{x_{1},x_{2},\ldots,x_{N}\}$ of $N \geq n$ given $(N+n)$-sums in $\Oh(\poly(t)N)$ time, and recall that, for every $i \in [N]$, $x_{i} = \sum_{j} \alpha^{i}_{j}g_{j}$ for some $\alpha^{i}_{1},\alpha^{i}_{2},\ldots,\alpha^{i}_{t}\in [0,N+n]$ such that $\sum_{j=1}^{t}\alpha^{i}_{j}\leq N+n$. We apply the preprocessing from \Cref{obs:compare}, so that after $\Oh(tN)$ preprocessing we can evaluate any such expression in $\Oh(t)$ time.
We will proceed in two phases. First, we will partition $X$ into groups of size $\log^{d} N$ in $\Oh(\poly(t)N)$ time, for some $d\in \mathbb{N}$
to be determined later.
The elements in each group will be in any order, but all elements in the $i$-th group will be smaller than the elements
in the $(i+1)$-th group. Second, we will sort the elements in every group in $\Oh(\poly(t)N)$ total time.
In both steps, the main idea is to ``approximate'' real numbers $\hat{g_{1}},\hat{g_{2}},\ldots,\hat{g_{t}}$ from Lemma~\ref{lem:real} with rational numbers $\tilde{g_{1}},\tilde{g_{2}},\ldots,\tilde{g_{t}}$.
This allows us to compute a rational approximation $\tilde{x}_{i}=\sum_{j}\alpha^{j}_{i}\tilde{g}_{j}$ of every $x_{i}=\sum_{j}\alpha^{i}_{j}g_{j}$.
Then, assuming that the rational numbers consist of only $\Oh(\poly(t)\log n)$ bits (which will be the case), we can sort
$\tilde{x}_{1},\tilde{x}_{2},\ldots,\tilde{x}_{N}$ in $\Oh(\poly(t)N)$ time in the Word RAM model.
To find the approximations, we apply \Cref{lem:whp}, and then use some tools from linear programming.

\paragraph{$\S $ Detailed description of the algorithm.} We now move to describing both phases of the algorithm in details. For each phase, we describe first what is done by our algorithm, and then the reasoning (and technical details) supporting the correctness and complexity of each step.

\paragraph{The first phase.}
We begin with choosing a subset $Y \subseteq X$ of size $m+1$ by always including the smallest and the largest element of $X$,
and independently including every other element of $X$ with probability $p=1/\log^{c}N$,
where $c\in \mathbb{N}$ will be fixed later so that $X$ is partitioned into groups of size $\Theta(\log^{d}n)$ with $d=2c+3$.
By a standard application of the multiplicative Chernoff bound, $m=\Theta(N/\log^{c}N)$ holds with high probability in $N$,
and from now we condition on this being the case.
We sort $Y$ with $\Oh(m\log m)$ comparisons, which takes $\Oh(tm\log m)=\Oh(tN)$ time,
as we can compare any two elements in $\Oh(t)$ time.
Let $x_{y_{0}} \lessg x_{y_{1}}\lessg \ldots \lessg x_{y_{m}}$ be the elements of $Y$ in the sorted order, and to avoid clutter assume that
$m$ is a multiple of $\log^{c+3}N$. 

Once $Y$ is defined and sorted, for $j=0,1,\ldots,m/\log^{c+3} N-1$, we define the sets:
\[ Z_{j} := \{ i \in [1,N] : x_{y_{j\cdot \log^{c+3} N}} \leqg x_{i} \text{ and } x_{i} \lessg x_{y_{(j+1)\log^{c+3} N}} \} . \]
As we have included the smallest and the largest element of $X$ in $Y$, $Z_{j}$s form a partition of $X$. Next, we analyse the size of every $Z_{j}$. The size of a single $Z_{j}$ can be bounded as follows.
Recall that every element of $X$ is included in $Y$ with probability $p$. Then, $|Z_{j}|$ is the number of trials
until obtaining $\log^{3}N$ successes in independent Bernoulli trials with probability $p$. A straightforward application of the
multiplicative Chernoff bound implies that among $p\cdot t/p$ trials we will have at least $t/2$ successes with probability at least
$1-\exp(-t/8)$. Thus, by setting $t=\Theta(\log^{3}N)$, we obtain that $|Z_{j}| = \Oh(\log^{3+c}N)$ holds with high probability in $N$.
From now on, we condition on $|Z_{j}|=\Oh(\log^{2c+3}N)$ holding
for all $j$, which happens with high probability in $N$.

We now focus on how to construct these sets $Z_j$. In order to effectively allocate the elements of $X$ into the sets $Z_{j}$ defined above, we will assign a rational number $\tilde{x}_{i}=p_{i}/q_{i}$ to every $x_{i}$, with the nominator and denominator consisting of $\Oh(\poly(t)\log n)$ bits, by first assigning such a rational number $\tilde{g}_{i}$ to every $g_{i}$, and then calculating $\tilde{x}_{i}=\sum_{j}\alpha^{i}_{j}\tilde{g_{j}}$.
In our model, arithmetical operations on numbers consisting of $\Oh(\poly(t)\log n)$ bits can be performed in
$\Oh(\poly(t))$ time, which allows us to compute the $\tilde{x}_{i}$s in $\Oh(\poly(t)N)$ time.
Then, we sort the $\tilde{x}_{i}$s in $\Oh(\poly(t)N)$ time using radix sort using the following lemma. We stress that sorting the $\tilde{x}_{i}$ will only approximate sorting their corresponding $x_{i}$s, and there will possibly
be elements $x_{i}$ and $x_{j}$ such that $x_{i} \lessg x_{j}$ but $\tilde{x_{i}} \geq \tilde{x_{j}}$. Still, we will be able to bound for how many
pairs of elements this can possibly happen.

\begin{lemma}
\label{lem:sorting}
Sorting $N$ rational numbers $p_{i}/q_{i}$, with each nominator and denominator consisting of $\Oh(\poly(t)\log n)$ bits,
can be done in $\Oh(\poly(t)N)$ time.
\end{lemma}

\begin{proof}
Firstly, we convert all rational numbers to integer numbers by replacing $p_{i}/q_{i}$ with $\lfloor p_{i}4^{b+1}/q_{i} \rfloor$.
The new numbers still consist of $\Oh(\poly(t)\log n)$ bits each, and each of them can be obtained in $\Oh(\poly(t))$ time
in our model.

Further, we argue that such a conversion does not change the relative order.
First, we observe that for two rational numbers $p/q$ and $p'/q'$ with the nominators and denominators consisting of $b$ bits,
where $p/q < p'/q'$, we have:
\[ \frac{p'}{q'} - \frac{p}{q} = \frac{p'\cdot q - p\cdot q}{q\cdot q'}  \geq \frac{1}{q'\cdot q'} \geq \frac{1} {2^{b}\cdot 2^{b}} = \frac{1}{4^{b}} , \]
so the difference between any two such numbers is either $0$ or at least $1/4^{b}$. 
Thus, if the numbers are different then increasing/decreasing both of them by at most $1/4^{b+1}$ does not change their relative order. Finally, we have the following inequalities:
\[  \frac{\lfloor p_{i}4^{b+1}/q_{i}\rfloor}{4^{b+1}} \leq \underbrace{\frac{p_{i} 4^{b+1}/ q_{i} } { 4^{b+1}}}_{= \: p_{i}/q_{i}} < \frac{\lfloor p_{i}4^{b+1}/q_{i}\rfloor+1}{4^{b+1}} \]
meaning that replacing $p_{i}/q_{i}$ with $\frac{\lfloor p_{i} 4^{b+1}/ q_{i}\rfloor } { 4^{b+1}}$ possibly increases it
by less than $1/4^{b+1}$, so replacing $p_{i}/q_{i}$ with $\lfloor p_{i} 4^{b+1} / q_{i}\rfloor$ indeed does not change the
relative order between the numbers. 

Finally, radix sorting proceeds to sort these integers in $\Oh(\poly(t))$ iterations, each taking $\Oh(N+n)$ time,
so $\Oh(\poly(t)N)$ overall. 
\end{proof}

Once the rational numbers $\tilde{x}_{i}$ are sorted, we will ultimately allocate the elements $x_i$ of $X$ to the sets $Z_j$, by partitioning the sorted list of rational numbers $\tilde{x}_i$ w.r.t. the rational numbers corresponding to the elements of $Y$ (which was already sorted). We will explain how this is exactly done in detail also below, once we clarify how the rational numbers $\tilde{x}_{i}$ are obtained. The rest of the description of the first phase of our algorithm describes how to implement these steps.

So, let us now explain how to assign a rational number $\tilde{x}_{i}=p_{i}/q_{i}$ to every $x_{i}$. We need to keep in mind our goal: namely, to ensure that
the rational numbers can be used instead of the original $x_{i}$s to partition $X$ into $Z_{j}$s. 
First, for every $i$ we define a point $h_{i}=(\alpha^{i}_{1},\alpha^{i}_{2},\ldots,\alpha^{i}_{t}) \in [0,N+n]^{t}$.
As mentioned earlier, by \Cref{lem:real} we can think that $g_{1},g_{2},\ldots,g_{t}\in \R$.
Then, we observe that $x_{i} \lessg x_{i'}$ is equivalent to $\sign(\langle g,h_{i}-h_{i'}\rangle) = -1$.
Thus, sorting the $x_{i}$s could be done by determining $\arrange_{H}(g)$, where $H=\{ h_{i}-h_{i'} : 1\leq i, i' \leq N\}$.
We note that clearly $|H| = N^{2}$, and by \Cref{lem:inference}, the inference dimension of $H$ is $\Oh(t\log (N+n))=\Oh(t\log N)$.
We choose a uniform subset $H'\subseteq H$ of size $s=tN/\log^{c} N$.
By \Cref{lem:whp}, $|H\setminus \infer(H',g)| = \frac{\Oh(t\log^{2} N)}{s} |H| = \Oh(N\log^{c+2} N)$
with high probability
in $N$, and from now we condition on this being the case.
Next, we want to resolve all label and comparison queries on $H'$. This is because, as $H'$ has been chosen uniformly at random, by \Cref{lem:whp} this will actually allow
us to resolve many comparison queries on the whole $H'$, which brings us closer to sorting the whole $X$.
Label queries can be resolved directly by computing, for every $h'\in H'$, $\sign(\langle g, h'\rangle)$.
To efficiently resolve all comparison queries, we sort the elements of $H'$ by declaring $h' \leq h''$ if and only if
$\langle g, h'\rangle \leqg \langle g,h''\rangle$, where $g=(g_{1},g_{2},\ldots,g_{t})$.
Sorting can be done with $\Oh(s\log s)=\Oh(tN)$ comparisons, so $\Oh(t^{2}N)$ time.
We denote the elements of $H'$ arranged according to this order by $h'_{1},h'_{2},\ldots,h'_{s}$, so that
$\langle g,h'_{1}\rangle \leqg \langle g,h'_{2}\rangle \leqg \cdots \leqg \langle g,h'_{s}\rangle$.
For every $i=1,2,\ldots,s-1$, we compute $\sign(\langle g, h'_{i}-h'_{i+1}\rangle)$ (observe that this is either $-1$ or $0$).
This takes $\Oh(ts)=\Oh(tN)$ time, and in fact implicitly resolves all comparison queries on $H'$.
Indeed, $\sign(\langle g, h'_{i}-h'_{j}\rangle)$ for $i<j$ is either $-1$ or $0$ due to sorting, and it is $0$
if and only if $\sign(\langle g, h'_{k}-h'_{k+1}\rangle) = 0$ for every $k=i,i+1,\ldots,j-1$.

Next, we create a system of $\Oh(tN/\log^{c} N)$ linear inequalities in $t$ variables $\tilde{g}_{1},\tilde{g}_{2},\ldots,\tilde{g}_{t}$.
Let $\tilde{g}=(\tilde{g}_{1},\tilde{g}_{2},\ldots,\tilde{g}_{t})$.
The goal of the inequalities is to encode the current information about $g$. In other words, for any $\tilde{g}$ that respects
all the inequalities, the behaviour of the algorithm should be the same when run with $g$ replaced by $\tilde{g}$.
First, for every $i=1,2,\ldots,m-1$, we add an inequality of the form $\langle \tilde{g}, h_{x_{y_{i}}}-h_{x_{y_{i+1}}}\rangle \leq -1$.
Second, for $i=1,2,\ldots,s$, if $\sign(\langle g, h'_{i}\rangle)=-1$ we add an inequality $\langle \tilde{g}, h'_{i}\rangle \leq -1$,
if $\sign(\langle g, h'_{i}\rangle) = 0$ we add inequalities $\langle \tilde{g}, h'_{i}\rangle \leq 0$ and $\langle \tilde{g}, h'_{i}\rangle \geq 0$,
and if $\sign(\langle g, h'_{i}\rangle)= 1$ we add an inequality $\langle \tilde{g}, h'_{i}\rangle \geq 1$.
Third, for every $i=1,2,\ldots,s-1$ if $\sign(\langle g, h'_{i}-h'_{i+1}\rangle) = -1$ we add an inequality $\langle \tilde{g}, h'_{i}-h'_{i+1}\rangle \leq -1$,
and if $\sign(\langle g, h'_{i}-h'_{i+1}\rangle) = 0$ we add inequalities $\langle \tilde{g}, h'_{i}-h'_{i+1}\rangle \leq 0$ and
$\langle \tilde{g}, h'_{i}-h'_{i+1}\rangle \geq 0$.
We now want to solve this system of inequalities.

We need a few (standard) definitions from combinatorial optimisation. We recommend the reader to consult the excellent
book by Gr{\"{o}}tschel, Lov{\'{a}}sz, and Schrijver~\cite{GLS1988} for the details.
A polyhedron $K$ is a set of solutions of a system of inequalities, or $K=\{ \vec{x} \in \R^{t} : A\vec{x} \leq b \}$.
The facet-complexity of $K$ is at most $\phi$ if the encoding length of each inequality is at most $\phi$, meaning
that each inequality has rational coefficients and the total length of the binary encodings of all rational numbers appearing
in the same inequality is at most $\phi$.
The strong non-emptiness problem for such a polyhedra $K$ amounts to finding some
$x\in K$ (where the coordinates of $x$ are described as rational numbers, so in fact $x\in \Q^{t}$),
or detecting that $K$ is empty. A strong separation oracle takes $y\in \Q^{t}$, checks if $y\in K$,
and if not returns a hyperplane that separates it from $K$, i.e. $c\in\R^{t}$ such that $\langle c,y\rangle > \max \{\langle c,x\rangle : x\in K \}$.
The following theorem is known, and determines our $c\in \mathbb{N}$.

\begin{theorem}[\cite{KHACHIYAN198053}]
\label{thm:emptiness}
For some $c\in \mathbb{N}$,
the strong non-emptiness problem for a $t$-dimensional polyhedra with facet-complexity $\phi$
given by a strong separation oracle can be solved in $\Oh(\poly(t)\phi^{c})$ time and calls to the oracle.
The length of the binary encoding of the found $x\in K$ (if any) is $\Oh(\poly(t)\phi)$.
~\footnote{The formulation of the first part of the theorem is from~\cite[Theorem 6.4.1]{GLS1988}.
The bound on the length of the binary encoding of $x$ follows by either carefully analysing the algorithm
or a black-box rounding procedure from~\cite[Theorem 6.2.13]{GLS1988}.}
\end{theorem}

Our system of linear inequalities clearly forms a $t$-dimensional polyhedron $K$. Further, its facet-complexity is
$\Oh(t\log (N+n))=\Oh(t\log N)$. We need to establish that $K$ is non-empty. 
Recall that, by \Cref{lem:real}, we can think that $g\in \R^{t}$.
It is not necessarily the case that $g\in K$, as we have encoded every constraint of the form $\sign(\langle g,h \rangle) = -1$ as
$\langle \tilde{g},h \rangle \leq -1$. However, for any $\alpha > 0$, $\sign(\langle g,h \rangle) = \sign(\langle \alpha g,h \rangle)$.
Thus, we can choose sufficiently large $\alpha > 0 $ so that in fact for each of the finitely many relevant constraints of the form
$\sign(\langle g,h \rangle) = -1$ we actually have $\langle g,h \rangle \leq -1$, and so $\alpha g \in K$, hence $K$ non-empty.
Thus, \Cref{thm:emptiness} allows us to find $\tilde{g}\in K$ with every coordinate
being a rational number with nominator and denominator consisting of $\Oh(\poly(t)\log N)=\Oh(\poly(t)\log n)$ bits,
assuming that we can implement an efficient strong separation oracle. A strong separation oracle can be implemented
by simply iterating over all the inequalities for a given $y$, checking if $\langle y, h\rangle \leq 0$ and returning
$h$ as the hyperplane that separates $y$ from $K$ otherwise. This takes $\Oh(t^{2} N/\log^{c} N)$ time per call,
making the overall time $\Oh(\poly(t)(t\log n)^{c}t^{2} N/\log^{c} N)=\Oh(\poly(t)N)$.

Having found $\tilde{g}=(\tilde{g}_{1},\tilde{g}_{2},\ldots,\tilde{g}_{t})$, we calculate $\tilde{x}_{i} = \langle \tilde{g},h_{i}\rangle$, for $i=1,2,\ldots,N$.
Every $\tilde{x}_{i}$ is a sum of $t$ rational numbers with nominators and denominators consisting of $\Oh(\poly(t)\log n)$,
thus can be represented as a rational number with nominator and denominator consisting also of $\Oh(\poly(t)\log n)$ bits.
This takes $\Oh(\poly(t)N)$ time in our model. Next, we sort $\tilde{x}_{1}, \tilde{x}_{2},\ldots,\tilde{x}_{N}$ with
radix sort in $\Oh(\poly(t)N)$ time as described earlier to obtain a sorted sequence
$\tilde{x}_{i_{1}} \leq \tilde{x}_{i_{2}} \leq \ldots \leq \tilde{x}_{i_{N}}$. Because of the inequalities of the form
$\langle \tilde{g}, h_{x_{y_{i+1}}}-h_{x_{y_{i}}}\rangle \leq -1$, for $i=1,2,\ldots,m-1$,
we have $\tilde{x}_{y_{1}} < \tilde{x}_{y_{2}} < \ldots < \tilde{x}_{y_{m}}$. That is, the order on the elements of $Y$
is the same when sorted using the original $x_{i}$s and their approximations $\tilde{x}_{i}$.
We partition all elements of $X$ into sets $\tilde{Z}_{j}$ as follows:
\[ \tilde{Z}_{j} := \{ i \in [1,N] : \tilde{x}_{y_{j\cdot \log^{c+3} N}} \leq \tilde{x}_{i} \text{ and } \tilde{x}_{i} < \tilde{x}_{y_{(j+1)\log^{c+3} N}} \} . \]
This takes only $\Oh(N)$ time as the elements are already sorted by the $\tilde{x}_{i}$s.
We claim that $\tilde{Z}_{j}$ is a reasonably good approximation of $Z_{j}$ in the following sense.
Consider $i \in \tilde{Z}_{j}$. We can check in constant time if $i\in Z_{j-1}, Z_{j}, Z_{j+1}$, and, if so, insert $i$
into the appropriate $Z_{j'}$, for $j'\in\{j-1,j,j+1\}$. Assume otherwise, by symmetry it is enough to consider the case where $x_{i} \lessg x_{y_{(j-1)\log^{c+3} N}}$.
In such a case, for every $i'\in [(j-1)\log^{c+3}N,j\cdot \log^{c+3}N)$ we have $\langle \tilde{g}, h_{y_{i'}}-h_{i}\rangle < 0$ but $\langle g, h_{y_{i'}}-h_{i}\rangle \greatg \zerog$.
Hence, such an $i$ contributes $\Omega(\log^{c+3} N)$ elements to the set $\{ h \in H : h \notin \infer(H', g) \}$.
However, we have conditioned on the size of this set being $\Oh(N\log^{c+2}N)$, so such a situation
can happen only for $\Oh(N/\log N)$ elements $x_{i}$. For each such $x_{i}$,
we binary search over the elements $x_{j\cdot \log^{c+3}N}$, for $j=1,2,\ldots,m/\log^{c+3} N$, to assign $i$ to the appropriate $Z_{j}$ in $\Oh(tN)$ total time. 

To summarise the first phase of the algorithm, in $\Oh(\poly(t)N)$ time we have partitioned $X$ into $Z_{j}$s such that
$|Z_{j}| = \Oh(\log^{2c+3}N)$ holds for every $j$ and the elements assigned to $Z_{j}$ are all strictly smaller than the elements assigned
to $Z_{j+1}$. Hence, it remains to sort every $Z_{j}$.

\paragraph{The second phase.}
The second phase proceeds similarly as the first phase.
We first observe that sorting every $Z_{j}$ could be done by determining
$\arrange_{H_{1}}(g)$, where $H_{1} = \{ h_{i}-h_{i'} : i,i' \in Z_{j} \text { for some } j\}$. We note that $|H_{1}| = \Oh(N\log^{2c+3}N)$, and again by \Cref{lem:inference} the inference dimension of $H_{1}$ is $\Oh(t\log (N+n))=\Oh(t\log N)$. 

We proceed as in the first phase, namely, we choose a uniform subset $H'_{1}\subseteq H_{1}$ of size $s'=tN/\log^{c} N$, and sort its elements by declaring $h' \leq h''$ if and only if $\langle g,h' \rangle \leqg \langle g,h''\rangle$.
We denote the elements of $H'_1$ arranged according to this order by $h''_{1},h''_{2},\ldots,h''_{s'}$. 

Then, we create a system of $\Oh(tN/\log^{c} N)$ linear inequalities in $t$ variables $\tilde{g}'_{1},\tilde{g}'_{2},\ldots,\tilde{g}'_{t}$ as follows. Let $\tilde{g}'=(\tilde{g}'_{1},\tilde{g}'_{2},\ldots,\tilde{g}'_{t})$.
For every $i=1,2,\ldots,s-1$, if $\sign(\langle g, h''_{i}-h''_{i+1}\rangle) = -1$, then we add an inequality $\langle \tilde{g}', h''_{i}-h''_{i+1}\rangle \leq -1$,
and if $\sign(\langle g, h''_{i}-h''_{i+1}\rangle) = 0$, then we add inequalities $\langle \tilde{g}', h''_{i}-h''_{i+1}\rangle \leq 0$ and $\langle \tilde{g}', h''_{i}-h''_{i+1}\rangle \geq 0$.
Note that, in contrast to the first phase, we now have only one type of inequalities. We solve this system of inequalities in time $\Oh(\poly(t)N)$ as in the first phase. 

Next, we calculate $\tilde{x}'_{i}=\langle \tilde{g}',h_{i}\rangle$, for $i=1,2,\ldots,N$ in $\Oh(\poly(t)N)$ time. 
Then, for every $Z_{j}$, we sort its elements $i\in Z_j$ by their corresponding $\tilde{x}'_{i}$. Very importantly, to do this step efficiently, we will do this sorting together for all $j$ by radix sort, i.e. we actually sort triples $(j,\tilde{x}'_{i},i)$ and so this takes $\Oh(\poly(t)N)$ time.

Now, consider a single $Z_{j} = \{ i_{1},i_{2},\ldots,i_k \}$, where $\langle \tilde{g}',h_{i_{1}}\rangle \leq \langle \tilde{g}',h_{i_{2}}\rangle \leq \ldots \leq \langle \tilde{g}', h_{i_{k}}\rangle$ because of the sorting we have just done. We claim that the number of inversions in the corresponding sequence $\langle g,h_{i_{1}}\rangle, \langle g,h_{i_{2}}\rangle, \ldots, \langle g,h_{i_{k}}\rangle$ is small. Indeed, every $a,b$ such that $\langle \tilde{g}', h_{i_{a}} \rangle \leq \langle \tilde{g}', h_{i_{b}}\rangle$ but
$\langle g, h_{i_{a}} \rangle \greatg \langle g, h_{i_{b}}\rangle$ corresponds to $h\in H_{1}$ such that $h\notin \infer(H'_{1}, g)$. Thus, by \Cref{lem:whp} the total number of inversions in all such sequences is only $\frac{\Oh(t\log^{2}N)}{s'}|H_{1}| = \Oh(\log^{3c+5}N)$ with high probability. This means that we can finish sorting all the sequences by e.g. insertion sort in $\Oh(t(N+\log^{3c+5}N))=\Oh(tN)$ time.

\paragraph{$ \S $ Summary.} 
After the second phase, we have correctly sorted the set $X$. Each phase works in $\Oh(tN)$ time, conditioned on some events
($m=\Theta(N/\log^{c}N)$, $|Z_{j}|=\Oh(\log^{2c+3}N)$ for every $j$,
$|H\setminus \infer(H',g)| = \Oh(N\log^{c+2} N)$,
 $|H_{1}\setminus \infer(H_{1}',g)| = \Oh(N\log^{3c+5} N)$)
 that happen with high probability in $N$.
 To obtain the statement of Lemma \ref{mainSortingLemma}, we observe
 that we can verify the correctness of the output (by checking if $x_{1}\leqg x_{2}\leqg \ldots \leqg x_{N}$) in $\Oh(tN)$
 time, and terminate the procedure if its running time exceeds $\Oh(tN)$ time. Then, with high probability in $N$ we obtain
 the correct output, and otherwise we restart the procedure.

This concludes the proof of Lemma \ref{mainSortingLemma}, and, as such, the main results of our paper, Theorems \ref{PathEnumWithSortingTheorem} and \ref{advancedEnumTheorem}, are now completely proven.

\section{Conclusions}
In this paper, we have approached the following ranked enumeration problem, called $\enumProbTransducer^f$, which was considered in Bourhis et al.~\cite{BourhisEtAl2021}: given a text document $\doc$ and an unambiguous cost transducer $\mathcal{T}$, and for $f$ being a polynomial function, enumerate the first $f(|\doc|)$ output tuples (obtained by running  $\mathcal{T}$ on $\doc$) in increasing order w.r.t. their weights. We have shown how this problem can be solved by reducing it to shortest path enumeration problem, for a weighted graph. Then, as a direct application of Eppstein's classical enumeration algorithm~\cite{Eppstein1998}, we obtained a ranked enumeration algorithm for $\enumProbTransducer^f$ such that, on input $\mathcal{T}$ and $\doc$, the preprocessing time is $\bigO(|\mathcal{T}| |\doc|)$ and the delay of outputting the $i^{th}$ element of the output, namely $s_i$, is $\bigO(|s_i| + \min\{ \log i, \log(|\mathcal{T}| |\doc|)\})$, for every $i \in [k]$.
Next, we designed a ranked enumeration algorithm for $\enumProbTransducer^f$ such that, on any input $(\mathcal{T}, \doc)$, the preprocessing time is still $\bigO(\poly(|\mathcal{T}|)|\doc|)$ and every output element $s$ has a delay of $\bigo(|s| \poly(|\mathcal{T}|))$, where the guarantees on the preprocessing time and of the delay hold with high probability in the size of the input and the size of the output generated so far. The main ingredient in obtaining this result, and also the main technical contribution of this paper, is a novel sorting method, allowing us to order efficiently the weights of the outputs produced by running $\mathcal{T}$ on $\doc$. 

\section*{Acknowledgments}

Pawe\l{} Gawrychowski is partially supported by the Polish National Science Centre (under grant number 2023/51/B/ST6/01505). 

Florin Manea is supported by the German Research Foundation (Deutsche Forschungsgemeinschaft, DFG) in the framework of the Heisenberg Programme – project number 466789228 (gef\"{o}rdert durch die Deutsche Forschungsgemeinschaft (DFG) – Projektnummer 466789228). 

Markus L. Schmid is supported by the German Research Foundation (Deutsche Forschungsgemeinschaft, DFG) – project number 522576760 (gef\"{o}rdert durch die Deutsche Forschungsgemeinschaft (DFG) – Projektnummer 522576760).

\bibliographystyle{plain}
\bibliography{biblio}

\appendix

\section{More Details About Eppstein's DAG of Shortest Paths}\label{sec:AppendixEppstein}

Let $G = (V, E, \weightName, \labelName)$ be a graph with edge weights from $\group$ and edge labels from $\monoid$. For any $u \in V$, a $u$-path is any path that starts in $u$. For nodes $u, v \in V$ such that there is a path from $u$ to $v$, we denote by $d(u, v)$ the \emph{distance} from $u$ to $v$, i.\,e., the length of a shortest path from $u$ to $v$. 

\medskip
\noindent\textbf{Eppstein's DAG of Shortest Paths.} Let $G = (V, E, \weightName, \labelName)$ be a graph with edge weights from $\group$ and edge labels from $\monoid$. Let $s, t \in V$ with $s \neq t$ be distinguished nodes of $G$. Since we are interested in all $s$-to-$t$ paths, we assume that every node is reachable from $s$, and that $t$ can be reached from every node (this is without loss of generality, since we can simply remove in linear time all nodes that violate this property). Let $T$ be a tree of the shortest paths from any node to $t$, i.\,e., $T$ is a spanning tree of $G$ with $t$ as its root (all edges are directed towards $t$) and with the property that for every $v \in V$, the $v$-to-$t$ path in $T$ is a shortest $v$-to-$t$ path in $G$. In particular, for every $v \in V$, we have now defined a unique shortest $v$-to-$t$-path in $G$, namely the $v$-to-$t$ path in $T$. Thus, in the following, we will talk about \emph{the shortest} $v$-to-$t$-path (in $G$) for any $v \in V$.

Every edge of $G$ that is not an edge of $T$ is called a \emph{sidetrack edge}. We observe that every possible $s$-to-$t$-path of $G$ is uniquely represented by its subsequence of sidetrack edges (i.\,e., the parts in between are uniquely defined by paths in $T$). Analogously, every sequence $(u_1, v_1), (u_2, v_2), \ldots, (u_k, v_k)$ of sidetrack edges such that, for every $i \in [k-1]$, $u_{i + 1}$ is reachable from $v_i$ in $T$, uniquely represents an $s$-to-$t$-path of $G$ (recall that, by definition, $u_1$ is reachable from $s$ and $t$ is reachable from $v_k$). We call sequences of sidetrack edges with this property \emph{valid}. This means that in order to represent all $s$-to-$t$-paths, it is sufficient to represent all sequences of sidetrack edges with this property.

For every edge $e = (u, v) \in E$, let $\delta(e)$ be the weight-difference between the shortest $u$-to-$t$-path and the $u$-to-$t$-path where we first take the edge $e$ and then follow the shortest $v$-to-$t$-path. More formally, $\delta(e) = \weightfunc{e} + d(v, t) - d(u, t)$. Obviously, $\delta(e)$ is non-negative. Intuitively speaking, $\delta(e)$ measures the additional distance caused by using the edge $e$ instead of just following the shortest path from $u$ (note that if $e$ is not a sidetrack edge, then these two paths are the same and therefore $\delta(e) = 0$).

We will now present a data structure that represents all $s$-to-$t$-paths of $G$ (along with their weights and labels) as a DAG. As mentioned above, we will represent $s$-to-$t$-paths by their subsequences of sidetrack edges.

An \emph{$(s, t)$-Eppstein-DAG} (\emph{of $G$}) is a directed acyclic graph $D_G$ with the following properties. Every node $x$ of $D_G$ represents a sidetrack edge $e(x) \in E$. For simplicity, we write $\delta(x) = \delta(e(x))$ for nodes $x$ of $D_G$. With every node $v \in V$, we identify a unique node $v^*$ of $D_G$ (note that, like all nodes of $D_G$, also $v^*$ represents some sidetrack edge of $G$, which, however, is not necessarily adjacent to $v$). There are two different kinds of edges in $D_G$: \emph{heap edges} and \emph{cross edges}. The heap edges are such that, for every $v \in V$, the nodes of $D_G$ reachable by heap edges from $v^*$ form a min-heap $H(v)$ of nodes representing the sidetrack edges of $G$ with their source on the shortest $v$-to-$t$-path; this min-heap $H(v)$ is ordered by the $\delta$-values and has $v^*$ as its root. In addition, every node $x$ of $D_G$ with $e(x) = (u, v)$ for some $u \in V$ has a cross edge from $x$ to $v^*$ (which is the root of the heap $H(v)$). Finally, we also add a distinguished node $r$ with a cross edge to $s^*$ (the root of the heap $H(s)$).

For convenience, let us denote $s$ by $v_1$ for the sake of the following explanation. An $r$-path in the DAG $D_G$ starts with the cross edge to $v_1^*$ (i.\,e., the root of the heap $H(v_1)$), then it traverses heap edges (possibly zero) in $H(v_1)$ until it jumps to the root $v^*_2$ of some other heap $H(v_2)$ by a cross edge, then it traverses some heap edges (possibly zero) in $H(v_2)$ until again it jumps to the root $v^*_3$ of another heap $H(v_3)$, and so on. Now let $x_1, \ldots, x_k$ be the nodes that are last visited in each of the traversed heaps, i.\,e., $x_1$ is last visited in $H(v_1)$, $x_2$ is last visited in $H(v_2)$, and so on. Or, equivalently, for every $i \in [k-1]$, $x_i$ is the source of the cross edge that leads to $H(v_{i+1})$ (i.\,e., with target $v^*_{i+1}$), and $x_k$ is the final node of the path. Note that if the path does not traverse any heap edges of some $H(v_i)$, then $x_i$ is the root $v^*_i$ of $H(v_i)$. We say that $x_1, x_2, \ldots, x_k$ are the nodes \emph{induced} by the $r$-path. By definition of $D_G$, this means that, for every $i \in [k-1]$, $e(x_i) = (u_i, v_{i+1})$ for some $u_i \in V$, and $e(x_k) = (u_k, v_{k+1})$ for some $u_k, v_{k+1} \in V$. Since each heap $H(v_i)$ stores exactly the sidetrack edges with their source on the shortest $v_i$-to-$t$-path, we can conclude that, for every $i \in \{2, 3, \ldots, k\}$, $u_i$ is reachable from $v_i$ in $T$. Consequently, the sequence $(u_1, v_2), (u_2, v_3), \ldots, (u_k, v_{k+1})$ is a valid sequence of sidetrack edges, and therefore represents an $s$-to-$t$-path. Finally, if we change the $r$-path, then we also change the induced nodes, which means that we represent a different $s$-to-$t$-path. 

On the other hand, let $P$ be an arbitrary $s$-to-$t$-path of $G$ with sidetrack edges $(u_1, v_2)$, $(u_2, v_3)$, $\ldots$, $(u_k, v_{k+1})$. Since $u_1$ is reachable from $v_1 = s$ in $T$, the edge $(u_1, v_2)$ is a sidetrack edge with its source on the shortest $v_1$-to-$t$-path; thus, $H(v_1)$ stores a node $x_1$ with $e(x_1) = (u_1, v_2)$. For every $i \in \{2, 3, \ldots, k\}$, since $u_i$ is reachable from $v_i$ in $T$, the edge $(u_{i}, v_{i+1})$ is a sidetrack edge with its source on the shortest $v_i$-to-$t$-path; thus, $H(v_i)$ stores a node $x_i$ with $e(x_i) = (u_i, v_{i+1})$. By induction, this means that there is an $r$-path in $D_G$ that induces the nodes $x_1, x_2, \ldots, x_k$ with $e(x_i) = (u_i, v_{i+1})$ for every $i \in [k]$. Moreover, if we change the $s$-to-$t$-path, then we also change the sequence of sidetrack edges, which means that we also change the corresponding $r$-path.

We conclude that there is a one-to-one correspondence between $s$-to-$t$-paths of $G$ and $r$-paths of $D_G$. Next, we add weights and labels to $D_G$ such that the $r$-paths of $D_G$ also describe the weights and labels of the corresponding $s$-to-$t$-paths of $G$.

We first add weights to the edges of $D_G$. For every heap edge $(x, y)$, we set $\weightfunc{(x, y)} = \delta(y) - \delta(x)$ (recall that nodes of $D_G$ represent edges of $G$ and therefore have $\delta$-values). Since $(x, y)$ is a heap edge, $\delta(x) \leqg \delta(y)$, which means that $\weightfunc{(x, y)}$ is non-negative. For every cross edge $(x, y)$, we set $\weightfunc{(x, y)} = \delta(y)$ (since $\delta$-values are non-negative, also in this case $\weightfunc{(x, y)}$ is non-negative). 

\begin{observation}\label{weightObservation}
For any $s$-to-$t$-path $P$ of $G$ and its corresponding $r$-path $P'$ of $D_G$, we have that $\weightfunc{P} = d(s, t) + \weightfunc{P'}$ (i.\,e., the $\delta$-values on the edges of $P'$ add up to the weight difference of the $s$-to-$t$ path $P$ and the shortest $s$-to-$t$-path). 
\end{observation}

Next, for every node $x$ of some heap $H(v)$ with $e(x) = (u, w)$, we define $\labelfunc{x}$ to be the label of the $v$-to-$u$-path of $T$. For every $w \in V$, let $\beta_w$ be the label of the shortest $w$-to-$t$-path. All the elements $\beta_w$ can be computed in linear time by traversing the tree $T$ starting in its root $t$. 

\begin{observation}\label{labelObservation}
Let $P$ be some $s$-to-$t$-path of $G$, let $P'$ be its corresponding $r$-path of $D_G$, and let $x_1, x_2, \ldots, x_k$ be the nodes induced by $P'$. Then $\labelfunc{P} = \labelfunc{x_1} \concatm \labelfunc{e(x_1)} \concatm \ldots \concatm \labelfunc{x_k} \concatm \labelfunc{e(x_k)} \concatm \beta_w$, where $w$ is the target of the edge $e(x_k)$.
\end{observation}

Eppstein showed in~\cite{Eppstein1998} that we can compute an $(s, t)$-Eppstein-DAG of linear size and with constant out-degree:

\begin{theorem}[Eppstein~\cite{Eppstein1998}]
Given a directed graph $G = (V, E)$ with edge weights from $\group$ and edge labels from $\monoid$, a tree $T$ of the shortest paths from any node of $G$ to $t$, and nodes $s, t \in V$, we can compute an $(s, t)$-Eppstein-DAG $D_G$ in time and space $\bigo(|G|)$ with the following properties. The DAG $D_G$ has maximum outdegree $4$, for every edge $e$ of $D_G$ the value $\weightfunc{e}$ is given as a pointer, and for every node $x$ of $D_G$ (except the node $r$) the value $\labelfunc{x}$ is given as a pointer.
\end{theorem}

Let $D_G$ be the $(s, t)$-Eppstein-DAG $D_G$ given by the above result.
We note that in our case, the input graph $G$ is a DAG, thus we can compute the shortest paths tree $T$ in $\bigo(|G|)$ time.
We next see that $D_G$ implicitly describes a min-heap $\mathcal{H}(G)$ of all $s$-to-$t$-paths of $G$ ordered by their weights. The root of the heap corresponds to the $r$-path of $D_G$ that consists in the single node $r$ (it therefore corresponds to the shortest $s$-to-$t$-path of $G$). For any node $q$ that represents some $r$-path $P$ of $D_G$, its children represent all $r$-paths obtained from $P$ by adding one edge. Since $D_G$ has maximum degree $4$, the degree of $\mathcal{H}(G)$ is at most $4$ as well. Since the edge weights of $D_G$ are non-negative, the weight of the path corresponding to any node of $\mathcal{H}(G)$ is never larger than the weight of the paths corresponding to its children. Thus, $\mathcal{H}(G)$ satisfies the heap-property. Next, we will see that we can navigate in this heap in constant time per step, and that we can do this in such a way that we can always access the weight and the label of the $s$-to-$t$-path corresponding to the current node of $\mathcal{H}(G)$.\par

Let us assume that we are at a heap node $q$ of $\mathcal{H}(G)$ that represents an $r$-path $P$ of $D_G$. Let us further assume that we have explicitly stored the nodes $x_1, x_2, \ldots, x_\ell$ induced by $P$, so that, for every $i \in [k]$, we have a pointer to the label $L_i =  \labelfunc{x_1} \concatm \labelfunc{e(x_1)} \concatm \ldots \concatm \labelfunc{x_i} \concatm \labelfunc{e(x_i)}$, and the weight $\weightfunc{P}$. In particular, this means that if $P$ corresponds to the $s$-to-$t$-path $P'$ of $G$, then, according to Observations~\ref{weightObservation}~and~\ref{labelObservation}, $\weightfunc{P'} = \weightfunc{P} + d(s, t)$ and $\labelfunc{P'} = L_k \concatm \beta_{w}$, where $w$ is the target of $e(x_k)$ (recall that $\beta_{w}$ is the label of the shortest $w$-to-$t$-path, which we have already precomputed). Consequently, we can access the weight and label of the $s$-to-$t$-path represented by $q$.\par
Appending a single edge to $P$ corresponds to moving from $q$ to one of its children. Removing the last edge of $P$ corresponds to moving from $q$ to its parent node. We have to explain how we can move down or up in $\mathcal{H}(G)$ while maintaining the list of induced nodes, the labels $L_i$ and the weight of  $P$. \par
If we move down by appending a heap edge $e = (y, z)$ to $P$, then we just have to replace $x_k$ with $z$, update $L_k$ to $L_{k-1} \concatm \labelfunc{z} \concatm \labelfunc{e(z)}$, and add $\weightfunc{e}$ to the current weight. If instead we move down by appending a cross edge $e = (y, z)$, then we have to append $x_{k+1} = z$ to the list of induced nodes, we have to set $L_{k + 1}$ to $L_k \concatm \labelfunc{z} \concatm \labelfunc{e(z)}$, and again we add $\weightfunc{e}$ to the current weight. On the other hand, when we move up in $\mathcal{H}(G)$ from $q$ to its parent $p$, then this corresponds to removing the last edge $e = (y, z)$ of $P$. If $e$ is a heap edge, then we have to set $x_k$ to $y$ and we have to change $L_k$ to $L_{k-1} \concatm \labelfunc{y} \concatm \labelfunc{e(y)}$. If $e$ is a cross edge, then we just have to remove $x_k$ and $L_k$. Moreover, in both case we have to subtract $\weightfunc{e}$ from the current weight. Obviously, each of these operation can be done in constant time.

\end{document}